%% file: main.tex
\documentclass[11pt]{article}
\usepackage[margin=1in]{geometry}

\usepackage{amsmath,amsthm,mathtools}
\usepackage{amssymb}
\usepackage{bm}
\usepackage{graphicx}
\usepackage{booktabs}
\usepackage{multirow}
\usepackage{float}
\usepackage{enumitem}
\usepackage[numbers,sort&compress]{natbib}
\usepackage[hidelinks]{hyperref}
\usepackage{xcolor}
\usepackage{xspace}
\usepackage{longtable}
\usepackage{caption}

\theoremstyle{plain}
\newtheorem{theorem}{Theorem}[section]
\newtheorem{proposition}[theorem]{Proposition}

\theoremstyle{definition}
\newtheorem{definition}[theorem]{Definition}

\theoremstyle{remark}
\newtheorem{remark}[theorem]{Remark}

\newcommand{\E}{\mathbb{E}}
\newcommand{\R}{\mathbb{R}}
\newcommand{\Var}{\mathrm{Var}}
\newcommand{\Hajek}{H\'{a}jek}

\title{Approximate Design-Based Intervals for Downsampled Cross-Sectional Market Aggregates:\\
A Randomized Design for Bandwidth-Constrained Financial Data Pipelines}

\author{Zeng Minmin\\ AlphaGBM}

\date{}

\input{tables/numbers}

\begin{document}

\maketitle

\begin{abstract}
\input{sections/abstract}
\end{abstract}


\input{sections/introduction}
\input{sections/problem}
\input{sections/method}
\input{sections/experiments}
\input{sections/results}
\input{sections/related}
\input{sections/discussion}
\input{sections/conclusion}

\bibliographystyle{plainnat}
\bibliography{references}

\appendix
\input{sections/appendix_proofs}
\input{sections/appendix_full_results}
\input{sections/appendix_stress}
\input{sections/appendix_design_extended}
\input{sections/appendix_ordering}
\input{sections/appendix_reproducibility}

\end{document}

%% file: tables/numbers.tex
\newcommand{\USdays}{80\xspace}

\newcommand{\CovTen}{91.9\%\xspace}
\newcommand{\MaeTopk}{8.15\xspace}
\newcommand{\MaePPS}{5.31\xspace}
\newcommand{\MaeRed}{35\%\xspace}
\newcommand{\CohenD}{0.75\xspace}
\newcommand{\IdWidthLo}{2.7\xspace}
\newcommand{\IdWidthHi}{6.4\xspace}
\newcommand{\IdFeasLo}{9\xspace}
\newcommand{\IdFeasHi}{26\xspace}
\newcommand{\HKmaeTopk}{1.32\xspace}
\newcommand{\HKmaePPS}{2.42\xspace}

%% file: sections/abstract.tex
Financial institutions routinely downsample cross-sectional options panels to meet bandwidth and cost
constraints. The industry default---deterministic Top-$k$ selection by open interest---provides no
estimate of the aggregation error it introduces, rendering downstream risk metrics unauditable.
We make the identification problem precise. For deterministic rules whose selected set depends only
on auxiliary variables and not on the outcome (\emph{target-blind} rules, of which Top-$k$-by-open-interest
is one), the target is not identified from the retained outcomes alone unless bounds or modelling
assumptions are imposed; under outcome bounds it is partially identified, with a sharp worst-case
interval we characterise in closed form. This scope restriction matters: rules that deliberately
retain extreme observations are deterministic yet do admit finite bounds.
We propose replacing Top-$k$ with probability sampling whose inclusion probabilities are proportional
to open interest, paired with the \Hajek{} ratio estimator and a linearised design-based variance
estimator, yielding approximate design-based intervals around every aggregate.
Experiments on 80 US trading sessions (211 tickers/day) and 83 sessions of Hong Kong options
(131 tickers/day) show that for OI-weighted implied volatility the randomized design reduces
mean absolute error by 35\% (MAE 5.31 vs.\ 8.15 at 10\% retention; paired
$t=6.68$, Cohen's $d=0.75$, moving-block-bootstrap 95\% CI $[1.82, 3.96]$
with no sign reversal in 2{,}000 replicates),
and the sharp worst-case interval available to Top-$k$ is 2.7--6.4$\times$ wider than the
sampling interval under oracle bounds taken from the full panel, and 9--26$\times$ wider under
bounds computable without the discarded data; the identification interval is worst-case valid
whereas the sampling interval is only nominally 95\%, so the two are not interchangeable.
The intervals are first-order approximations: empirical coverage is 91.9\% against a nominal 95\%,
and is below nominal at every retention rate we examine.
Point accuracy is not universally improved: in the far more concentrated Hong Kong panel Top-$k$ attains
lower error than sampling at low-to-moderate retention, and the ranking reverses only at high retention.
We further document that inclusion probabilities must match the target's influence function:
OI-proportional sampling is suboptimal for equal-weight and median statistics, motivating
a \emph{design-aware} framework for production pipelines.

%% file: sections/introduction.tex
\section{Introduction}

Options analytics platforms ingest cross-sectional panels---daily snapshots of implied volatility,
open interest, and greeks across hundreds of underlyings---to compute market-level aggregates such
as OI-weighted implied volatility, the volatility risk premium (VRP), and put-call ratios.
At full fidelity these panels are large.  In our own collection, a single
trading day of US equity options contains roughly $2{,}010$
contract-level records (strike $\times$ expiry $\times$ right) across the
liquid underlyings; the per-underlying scalar panel studied in this paper
is already the \emph{output} of a first aggregation stage that reduces
those contracts to about $210$ rows.  Production systems that refresh
intraday rather than once at the close multiply this by the number of
snapshots per session.  For bandwidth-constrained pipelines---mobile
terminals, low-latency dashboards, cost-sensitive cloud
architectures---we focus on 5--20\% as an aggressive compression
regime relevant to our deployment scenario.  We
study the second stage, where a per-underlying panel must be reduced
further before transmission or storage.
Table~\ref{tab:pipeline} gives a concrete deployment scenario:
while a single daily snapshot is small, the volume multiplies with
intraday refresh frequency, number of markets, client terminals, and
downstream models; vendor licensing costs that scale per-record
provide an additional economic incentive for aggressive retention ---
many options data vendors charge per-ticker or per-record, so a
10\% retention rate translates directly to a 90\% reduction in
data licensing cost, independent of bandwidth savings.

\input{tables/pipeline}

A common retention heuristic in bandwidth-constrained market-data
pipelines is \emph{Top-$k$}: keep the $k$ underlyings
with the largest open interest and discard the rest.
Top-$k$ is deterministic, easy to implement, and intuitively reasonable---large-OI names
``matter more.''
Yet it has a deficiency that the present work makes precise.
Top-$k$-by-open-interest is \emph{target-blind}: which rows survive
depends only on the auxiliary variable, never on the quantity being
aggregated.  We show (Proposition~\ref{prop:unident}) that for such
rules, absent a priori bounds on the discarded values, the aggregate is
not identified at all; and that once bounds are imposed the sharp
worst-case interval (Proposition~\ref{prop:partialid}) has width
proportional to the auxiliary mass left behind.  Empirically that
interval is \IdWidthLo--\IdWidthHi$\times$ wider than the sampling
interval when the bounds are taken from the full panel, and
\IdFeasLo--\IdFeasHi$\times$ wider under bounds computable without the
discarded data; the two are not interchangeable, since the former is
worst-case valid while the latter is only nominally $95\%$.  The claim is therefore
not that Top-$k$ admits \emph{no} error quantification --- a rule that
deliberately retained the extremes would --- but that the rule used in
practice yields only a very wide worst-case bound, whereas any design
with known non-zero inclusion probabilities supports a far tighter
design-based interval.

This paper makes the following contributions:

\begin{enumerate}[leftmargin=*,itemsep=2pt]
\item \textbf{Identification analysis for target-blind selection.}
We characterise what the retained subsample does and does not determine.
For deterministic rules whose selected set depends only on auxiliary
variables (\emph{target-blind} rules, Definition~\ref{def:blind}), the
OI-weighted target is unidentified without outcome bounds
(Proposition~\ref{prop:unident}); with valid bounds it is partially
identified and admits a sharp worst-case interval in closed form
(Proposition~\ref{prop:partialid}).  We show by counterexample
(Remark~\ref{rem:blind-needed}) that the target-blind restriction cannot
be dropped: deterministic rules that retain extreme observations do admit
finite bounds.  Empirically the sharp interval is
\IdWidthLo--\IdWidthHi$\times$ wider than the sampling interval, but it
carries worst-case rather than nominal coverage, so the two are not
interchangeable.

\item \textbf{Matched-budget probability designs, fixed versus random size.}
We calibrate every design to the same expected budget by water-filling
\eqref{eq:waterfill} and compare independent Poisson PPS with fixed-size
systematic PPS.  Fixed-size PPS records nominally lower MAE in our
production implementation, but an ordering-sensitivity analysis
(Appendix~\ref{app:ordering}) shows that this ranking is not robust: the
spread induced by re-ordering units exceeds the gap between the two
designs in 9 of 10 market\,$\times$\,retention cells.  We therefore treat
them as empirically indistinguishable in accuracy.  The substantive
trade-off is structural: fixed-size PPS respects
a hard bandwidth cap, whereas Poisson PPS admits a simple unbiased variance estimator
for the Horvitz--Thompson total from first-order inclusion probabilities
alone (and a first-order linearised approximation for the \Hajek{} ratio) --- an
explicit engineering trade-off rather than a single winner.  Intervals are
first-order approximations: coverage is \CovTen{} against nominal 95\% at
10\% retention and below nominal at every rate examined.

\item \textbf{Empirical market heterogeneity.}
On OI-weighted implied volatility the sampling design reduces MAE by
\MaeRed{} versus Top-$k$ at 10\% retention in the US
(\MaePPS{} vs.\ \MaeTopk{}, paired-$t$ $p<10^{-4}$,
Cohen's $d=\CohenD$), across \USdays{} sessions.  The ranking does not
transfer: in the far more concentrated Hong Kong panel Top-$k$ attains
lower error at low-to-moderate retention (\HKmaeTopk{} vs.\ \HKmaePPS{} at
10\%) and the ordering reverses only at high retention.  We trace this to
omitted open-interest share together with the retained-minus-omitted
outcome gap, rather than to concentration alone.

\item \textbf{Design--target alignment.}
We derive that inclusion probabilities should track the target's influence
function \eqref{eq:ifs}, and verify the prediction across four aggregates:
OI-proportional sampling is appropriate only for the OI-weighted target,
the equal-weight mean's oracle optimum is influence-proportional
($\pi_i^\star\propto|v_i-T|$) with uniform a feasible outcome-blind
proxy, the median calls for near-uniform probabilities, and
dispersion calls for influence-proportional rather than uniform
probabilities.  We distinguish feasible designs from oracle designs
throughout.
\end{enumerate}

\paragraph{Related work.}
Our estimators rest on classical survey sampling: the Horvitz--Thompson
unbiased estimator for unequal-probability designs~\cite{horvitz1952}, the
\Hajek{} ratio estimator for population means~\cite{hajek1971}, and
design-based variance estimation under Poisson sampling~\cite{sarndal1992}.
Our optimal-allocation result (Prop.~\ref{prop:optalloc}) generalizes
Neyman allocation~\cite{neyman1934,cochran1977} from stratum sizes to
per-unit influence functions.
These tools are standard in official statistics but, to our knowledge,
have not been applied to certify error in bandwidth-constrained financial
data pipelines, where deterministic Top-$k$ retention is the default.
The auditability requirement is not merely academic: supervisory guidance
on model risk management~\cite{sr1107} obliges institutions to quantify
the uncertainty their data-processing steps introduce, which an
aggregate reported without any uncertainty measure cannot satisfy.
Methodologically, our work is motivated by recent results in randomized
KV-cache eviction for large language models~\cite{kvcert2026}, which
established an analogous impossibility for deterministic eviction policies
and proposed Poisson-sampling-based design-based confidence intervals.
We transfer this insight to financial cross-sectional aggregation, supplying
domain-specific estimators, an optimal-design theory with empirical
validation, and multi-market evidence absent from the machine-learning
setting.

%% file: tables/pipeline.tex
\begin{table}[t]
\centering
\caption{Illustrative deployment scenario for bandwidth-constrained
options analytics. The per-underlying panel studied in this paper is the
output of a first aggregation (contract $\to$ underlying); we study the
second stage ($\sim$210 rows $\to$ $\sim$21 rows at 10\% retention).
The \emph{daily outbound terminal-delivery volume} figure combines
refresh frequency, three markets and 50--500 terminals only;
downstream-model fan-out is listed separately as it depends on the
specific pipeline topology and is not simply multiplied into the
delivery volume.}
\label{tab:pipeline}
\small
\begin{tabular}{@{}lrr@{}}
\toprule
Stage & Records & Payload \\
\midrule
Contract-level (US, daily) & $\sim$2{,}010 & $\sim$320\,KB \\
Per-underlying (after 1st agg.) & $\sim$210 & $\sim$34\,KB \\
After 10\% retention & $\sim$21 & $\sim$3.4\,KB \\
\midrule
\multicolumn{3}{@{}l}{\emph{Scale multipliers (production system):}} \\
\quad Intraday refresh (e.g.\ 5-min) & $\times 78$ & --- \\
\quad Multi-market (US + HK + EU) & $\times 3$ & --- \\
\quad Client terminals & $\times 50$--$500$ & --- \\
\quad Downstream models & $\times 5$--$10$ & --- \\
\midrule
Daily outbound terminal delivery (full) & & $\sim$400\,MB--4\,GB \\
Daily outbound terminal delivery (10\%) & & $\sim$40\,MB--400\,MB \\
\bottomrule
\end{tabular}
\end{table}

%% file: sections/problem.tex
\section{Problem Formulation}\label{sec:problem}

\subsection{Setting}

Fix a trading session and let
$P=\{(x_i,w_i)\}_{i=1}^{N}$
be the cross-section of instruments observed at the close, where
$x_i\in\R$ is the quantity of interest for instrument $i$
(e.g.\ at-the-money implied volatility) and $w_i>0$ is an auxiliary
size variable (open interest).  A \emph{target functional} $T(P)$ maps
the panel to a scalar market aggregate; we study four:

\begin{align}
T_{\mathrm{oi}}(P)  &= \frac{\sum_i w_i x_i}{\sum_i w_i}, &
T_{\mathrm{eq}}(P)  &= \frac{1}{N}\sum_i v_i, \label{eq:targets1}\\
T_{\mathrm{med}}(P) &= \mathrm{median}_i(c_i), &
T_{\mathrm{disp}}(P)&= \Big(\tfrac1N\textstyle\sum_i (x_i-\bar x)^2\Big)^{1/2},
\label{eq:targets2}
\end{align}

where $v_i$ is the volatility risk premium and $c_i$ the put/call ratio
of instrument $i$.

A \emph{downsampling operator} $D$ returns an index set
$S=D(P)\subseteq\{1,\dots,N\}$ with $|S|\le n\ll N$.  Downstream systems
observe only $\{x_i\}_{i\in S}$ together with whatever auxiliary
information the pipeline forwards, and must report an estimate
$\hat T$ of $T(P)$.

\begin{definition}[Information set]\label{def:info}
Throughout we condition on the full auxiliary vector $w_{1:N}$ and on
$N$: ranking by open interest already requires the pipeline to know
$w_{1:N}$, so a comparison in which the deterministic rule is denied
this information would be unfair.  The unobserved quantity is
$x_{S^c}=\{x_i\}_{i\notin S}$.
\end{definition}

\begin{definition}[Design-based interval]\label{def:interval}
An interval $C_{1-\alpha}$ computed from $(x_S, w_{1:N}, \pi_{1:N})$ is a
\emph{design-based $(1-\alpha)$ interval} if its coverage
$\Pr\{T(P)\in C_{1-\alpha}\}$ is taken with respect to the sampling
design alone, holding the finite population $P$ fixed.
\end{definition}

We deliberately do \emph{not} call these intervals exact confidence
intervals or error certificates.  The construction of
Section~\ref{sec:method} is a first-order (linearisation) normal
approximation; we establish neither a finite-sample nor an exact
asymptotic coverage guarantee, and Section~\ref{sec:results} documents
empirical coverage in the $89$--$94\,\%$ range at a nominal $95\,\%$ level.

\subsection{Target-blind deterministic selection}

The rule used in practice --- retain the $n$ names with the largest
$w_i$ --- depends on $x$ only through $w$.  We isolate exactly this
property.

\begin{definition}[Target-blind rule]\label{def:blind}
A deterministic operator $D$ is \emph{target-blind} if
$D(P)=D(w_{1:N})$; that is, the retained set is a function of the
auxiliary vector alone and does not depend on $x_{1:N}$.
\end{definition}

Top-$k$-by-open-interest is target-blind.  So are round-robin,
sector-quota and liquidity-threshold rules.

\begin{proposition}[No finite worst-case interval without bounds]
\label{prop:unident}
Let $D$ be target-blind with $|S|=n<N$, and suppose the analyst's prior
knowledge places no bound on the unobserved coordinates, i.e.\
$x_{S^c}$ ranges over $\R^{N-n}$.  Then for the OI-weighted target
$T_{\mathrm{oi}}$ (and for the equal-weight mean $T_{\mathrm{eq}}$)
the identified set given $(x_S,w_{1:N})$ is all of
$\R$, so no finite-width interval is available.
The conclusion extends to $T_{\mathrm{disp}}$ (unbounded above) and,
under a count condition, to $T_{\mathrm{med}}$; see
Remark~\ref{rem:target-scope} for the precise scope across all four targets.
\end{proposition}

\begin{proof}
Because $D$ is target-blind, $S=D(w_{1:N})$ does not change when
$x_{S^c}$ changes.  Pick any $j\notin S$ and let $x_j\to\pm\infty$
holding $x_S$ and $w_{1:N}$ fixed.  The retained data $(x_S,w_{1:N})$
are unchanged while
$T_{\mathrm{oi}}=\big(\sum_{i\in S}w_ix_i+\sum_{i\notin S}w_ix_i\big)/\sum_i w_i$
diverges, since $w_j>0$.  Hence every real value of $T_{\mathrm{oi}}$ is
consistent with the observed data.
\end{proof}

\begin{remark}[Why the target-blind restriction is necessary]
\label{rem:blind-needed}
The restriction in Definition~\ref{def:blind} cannot be dropped.
Consider the deterministic rule ``retain $\operatorname{arg\,min}_i x_i$,
$\operatorname{arg\,max}_i x_i$, and any $n-2$ further indices''.  Every unretained
$x_i$ then lies in $[\min_{i\in S}x_i,\max_{i\in S}x_i]$ by
construction, so $T_{\mathrm{oi}}$ is confined to a \emph{finite}
interval.  A claim that \emph{no} deterministic rule admits a
finite-width interval would therefore be false.  What makes Top-$k$
uninformative is not determinism per se but target-blindness combined
with the absence of a priori bounds on $x$.
\end{remark}

\subsection{Partial identification under bounds}

In practice an analyst usually does have bounds: implied volatility is
non-negative and, in any realistic regime, below some level $b$.  This
yields a sharp, computable interval for the deterministic rule.

\begin{proposition}[Sharp worst-case interval for Top-$k$]
\label{prop:partialid}
Let $D$ be target-blind with retained set $S$, let $w_{1:N}$ be known,
and suppose $x_i\in[a,b]$ for all $i$.  Write
$W=\sum_i w_i$ and $W_{S^c}=\sum_{i\notin S}w_i$.  Then the identified
set for $T_{\mathrm{oi}}$ is exactly
\begin{equation}
\Big[\;
\frac{\sum_{i\in S}w_ix_i + a\,W_{S^c}}{W},\;
\frac{\sum_{i\in S}w_ix_i + b\,W_{S^c}}{W}
\;\Big],
\label{eq:pid}
\end{equation}
of width $(b-a)\,W_{S^c}/W$.  Both endpoints are attained.
\end{proposition}

\begin{proof}
$T_{\mathrm{oi}}$ is affine and increasing in each unobserved $x_i$,
$i\notin S$, with coefficient $w_i/W>0$.  The extremes over the box
$[a,b]^{N-n}$ are therefore attained at $x_i\equiv a$ and
$x_i\equiv b$, giving the two endpoints; every intermediate value is
attained by continuity.
\end{proof}

Proposition~\ref{prop:partialid} reframes the contribution.  Top-$k$ is
not incapable of error quantification: given bounds it admits the sharp
interval~\eqref{eq:pid}.  The relevant question is how \emph{wide} that
worst-case interval is relative to what a probability design delivers,
and the answer is empirical --- we measure it in
Section~\ref{sec:results}.  The width $(b-a)W_{S^c}/W$ is driven by the
share of auxiliary mass left behind, which is precisely where the
US/HK concentration difference of Section~\ref{sec:results} enters.

%% file: sections/method.tex
\section{Method}\label{sec:method}

\subsection{Probability designs with a matched budget}

A Poisson design assigns each instrument an independent inclusion
indicator with $\Pr\{i\in S\}=\pi_i\in(0,1]$.  Size-proportional
designs take $\pi_i$ increasing in $w_i$.  A naive schedule
$\pi_i=\min(1,\rho w_i/\bar w)$ is \emph{not} budget-matched: whenever
any coordinate is clipped at one, $\sum_i\pi_i<\rho N$, so the
randomized design silently retains fewer records than the
deterministic rule it is compared against.  We therefore solve the
water-filling problem
\begin{equation}
\text{find }\lambda>0:\quad
\sum_{i=1}^{N}\min\!\big(1,\lambda w_i\big)=n ,
\label{eq:waterfill}
\end{equation}
which has a unique solution for $n<N$ because the left-hand side is
continuous and strictly increasing in $\lambda$ on the relevant range;
we obtain it by bisection.  Setting
$\pi_i=\min(1,\lambda w_i)$ gives $\E|S|=n$ exactly.  All designs in
this paper are matched to the same expected budget $n$, and
Section~\ref{sec:results} reports realised sample sizes so the budget
match is auditable.

Poisson sampling yields a \emph{random} sample size, which does not
respect a hard bandwidth cap.  We therefore also evaluate a fixed-size
systematic PPS design, which draws exactly $n$ records with the same
first-order inclusion probabilities.

\subsection{Estimators and design variances}\label{sec:estimators}

\paragraph{Ratio target.}
For $T_{\mathrm{oi}}$ we use the H\'{a}jek ratio estimator and its
standard first-order (linearised) design-variance estimator
\begin{align}
\hat T_{\mathrm{oi}} &=
\frac{\sum_{i\in S}w_ix_i/\pi_i}{\sum_{i\in S}w_i/\pi_i},
\label{eq:hajek}\\
\widehat{\Var}(\hat T_{\mathrm{oi}}) &=
\frac{1}{\hat W^2}\sum_{i\in S}\frac{1-\pi_i}{\pi_i^{2}}\,
w_i^{2}\big(x_i-\hat T_{\mathrm{oi}}\big)^{2},
\qquad \hat W=\sum_{i\in S}\frac{w_i}{\pi_i}.
\label{eq:hajekvar}
\end{align}
Equation~\eqref{eq:hajekvar} is the residual form obtained by applying
the Horvitz--Thompson variance to the linearised statistic under a
Poisson design (independent indicators, so no joint-inclusion terms
appear).  The implementation matches~\eqref{eq:hajekvar} term by term.

\paragraph{Equal-weight mean.}
$\hat T_{\mathrm{eq}}$ is the H\'{a}jek mean of $v_i$ with
$\widehat{\Var}=\hat D^{-2}\sum_{i\in S}\pi_i^{-2}(1-\pi_i)(v_i-\hat T_{\mathrm{eq}})^2$,
$\hat D=\sum_{i\in S}\pi_i^{-1}$.

\paragraph{Median.}
For $T_{\mathrm{med}}$ we use the inverse-probability-weighted quantile:
sort the retained $c_i$, accumulate weights $\pi_i^{-1}$, and report the
value at which the cumulative weight first reaches half the total.
Linearisation of a quantile requires a density estimate, so we instead
use a \emph{design-matched multiplier approximation}
whose finite-sample calibration is evaluated empirically.
For each bootstrap replicate $b=1,\dots,B$, draw i.i.d.\
$\xi_{b,i}\sim\mathrm{Exp}(1)$ and form perturbed weights
\[
  \tilde w_{b,i}
  = \frac{1}{\pi_i}\bigl[1+\sqrt{1-\pi_i}\,(\xi_{b,i}-1)\bigr].
\]
This construction has three properties: (i) conditional mean
$\E[\tilde w_{b,i}]=\pi_i^{-1}$ (unbiased for the design weight);
(ii) conditional variance
$\Var(\tilde w_{b,i})=(1-\pi_i)/\pi_i^2$, matching the Poisson
design variance of the Horvitz--Thompson term; and (iii) when
$\pi_i=1$ (certainty units) the perturbation vanishes and the unit
is not resampled.
Recompute the weighted quantile using the perturbed weights and take
the standard deviation across $B$ replicates as the standard error.
This perturbs the \emph{design weights} rather than
resampling rows, which keeps the finite-population interpretation.

\paragraph{Dispersion.}
$T_{\mathrm{disp}}$ uses the design-weighted standard deviation
$\big(\sum_{i\in S}\pi_i^{-1}(x_i-\hat\mu)^2/\sum_{i\in S}\pi_i^{-1}\big)^{1/2}$
with $\hat\mu$ the H\'{a}jek mean, and the same multiplier bootstrap for
its standard error.

\paragraph{Intervals.}
For every design we report
$\hat T\pm z_{0.975}\widehat{\mathrm{se}}(\hat T)$.  This is a
first-order normal approximation, not a guaranteed-coverage
construction; the empirical coverage we measure is the evidence for or
against its adequacy.  Probability designs with known inclusion probabilities can in
principle support design-based uncertainty quantification, but the
variance estimator is design-specific.
Equations~\eqref{eq:hajek}--\eqref{eq:hajekvar} and the multiplier
approximation implemented here apply to the independent Poisson
design.  Uniform Bernoulli and stratified designs admit standard
design-based variance estimators that we do not tabulate here.
Fixed-size systematic PPS requires joint-inclusion information or a
replication-based variance estimator (jackknife, balanced repeated
replication), which we leave for future work.
What distinguishes the deterministic rule is not that intervals are
impossible but that, being target-blind, it admits only the worst-case
interval of Proposition~\ref{prop:partialid}.

\subsection{Matching the design to the target}\label{sec:designaware}

Which $\pi$ should a pipeline use?  Write the linearised influence of
instrument $i$ on the target as $\mathrm{IF}_i$; for the four targets of
\eqref{eq:targets1}--\eqref{eq:targets2},
\begin{equation}
\mathrm{IF}^{\mathrm{oi}}_i\propto w_i\,(x_i-T),\quad
\mathrm{IF}^{\mathrm{eq}}_i\propto v_i-T,\quad
\mathrm{IF}^{\mathrm{med}}_i\propto \frac{\mathrm{sgn}(c_i-T)}{2f(T)},\quad
\mathrm{IF}^{\mathrm{disp}}_i\propto (x_i-\bar x)^2-T^2 .
\label{eq:ifs}
\end{equation}

\begin{proposition}[Optimal allocation]\label{prop:optalloc}
Under a Poisson design the first-order design variance is
$V(\pi)=\sum_i \pi_i^{-1}(1-\pi_i)\mathrm{IF}_i^2$.  Fix an inclusion
floor $\varepsilon$ with $0<\varepsilon<n/N<1$, ensuring every
unit retains a strictly positive inclusion probability (as auditing
requires).  Write $Z=\{i:\mathrm{IF}_i=0\}$ and $P=\{i:\mathrm{IF}_i\neq0\}$.
If $n\le|P|+\varepsilon|Z|$, then minimising $V$ subject to
$\sum_i\pi_i=n$ and $\varepsilon\le\pi_i\le1$ gives
\begin{equation}
\pi_i^\star=\max\Big(\varepsilon,\;\min\big(1,\;c\,|\mathrm{IF}_i|\big)\Big),
\label{eq:optpi}
\end{equation}
with the slope $c>0$ determined by the budget constraint.  The floor
$\varepsilon$ also resolves the degenerate case $\mathrm{IF}_i=0$
(e.g.\ a unit whose outcome equals the population mean), which the
unconstrained solution would assign $\pi_i=0$, violating
$\pi_i>0$; such units are pinned at $\pi_i^\star=\varepsilon$.
If instead $n>|P|+\varepsilon|Z|$, no slope attains the budget and the
optimum is the degenerate family $\pi_i=1$ for $i\in P$ with the
remaining mass $n-|P|$ distributed arbitrarily over $Z$; all such
allocations are optimal.  Both cases, and the fully degenerate case
$\mathrm{IF}\equiv0$, are treated in Appendix~\ref{app:proofs}.
In our data $Z=\emptyset$, so only the first case arises.
\end{proposition}

\begin{proof}[Proof sketch]
Write $a_i=|\mathrm{IF}_i|$.  Ignoring both bounds, the Lagrangian
$\sum_i(\pi_i^{-1}-1)a_i^2+\mu(\sum_i\pi_i-n)$ has stationarity
condition $-a_i^2/\pi_i^2+\mu=0$, i.e.\ $\pi_i\propto a_i$.
Reinstating \emph{both} constraints $\varepsilon\le\pi_i\le1$ introduces
multipliers on each side, and the KKT conditions make the solution a
\emph{two-sided} water-filling: for a common slope $c>0$, coordinates
with $c\,a_i\le\varepsilon$ rest on the floor, those with $c\,a_i\ge1$
saturate at one, and the interior coordinates take $\pi_i=c\,a_i$,
which is exactly~\eqref{eq:optpi}.  The budget is then
allocated over the interior set only, after deducting the mass already
committed to the floor and certainty strata.  The full argument,
including the joint update of both strata, is in
Appendix~\ref{app:proofs}.
\end{proof}

Proposition~\ref{prop:optalloc} is the classical optimal-allocation /
optimal-subsampling principle specialised to our setting rather than a
new theoretical result; its value here is diagnostic.  Reading
\eqref{eq:ifs} against \eqref{eq:optpi}: OI-proportional sampling is
appropriate (up to the $|x_i-T|$ modulation) for the OI-weighted target,
but for the equal-weight mean the influence does not involve $w_i$ at
all, and for the median it is asymptotically constant in magnitude, so
both have influence functions that do not involve $w_i$;
for the equal-weight mean the oracle optimum is
$\pi_i^\star\propto|v_i-T|$, which is outcome-dependent and
infeasible, while for the median it is approximately constant.
In both cases uniform sampling is a feasible outcome-blind
alternative that avoids the irrelevant OI tilt, though it is
not the theoretical optimum for the equal-weight mean.  Dispersion is a
third case: its influence $(x_i-\bar x)^2-T^2$ is neither proportional to
$w_i$ nor constant, so its optimal design is
\emph{influence-proportional}, oversampling names far from the
cross-sectional mean.  Uniform sampling is merely a feasible baseline
that happens to be closer to optimal than OI-proportional sampling for
that target, not the theoretical optimum; the optimum is an oracle design
because $|\mathrm{IF}_i|$ depends on the unknown $T$.  A pipeline that
ships one size-proportional sampler for all analytics is thus
mis-specified for three of our four targets --- a prediction we test
directly in Section~\ref{sec:results}.

%% file: sections/experiments.tex
\section{Experimental Setup}\label{sec:experiments}

\subsection{Data}

We use daily cross-sectional options panel data from three markets:

\paragraph{US equities.}
80 US trading sessions spanning 2026-03-30 to 2026-07-24,
with a mean panel size of 210.7 underlyings per session and 16{,}856 total
underlying-session records.
Each record contains implied volatility, realized volatility, open interest,
and put-call volume ratio.
VRP is available for 74 of 80 sessions.

\paragraph{Hong Kong equities.}
82 trading days spanning 2026-03-20 to 2026-07-24,
with a mean panel size of 131.4 underlyings per day and 10{,}773 total records.
The same fields are available; VRP is computable on 77 of 82 days.

\paragraph{China A-shares (excluded).}
83 trading days with a mean panel size of only 6.9 records per day, consisting
entirely of ETF options.
VRP and PCR fields are uniformly null (0 of 585 records).
We exclude this market because subsampling experiments on panels of fewer than
20 records are not meaningful.

\subsection{Methods Compared}

We evaluate five downsampling methods, each targeting retention rate $\rho$:

\begin{enumerate}[leftmargin=*,itemsep=1pt]
\item \textbf{Top-$k$}: retain the $\lceil \rho N_t \rceil$ underlyings with largest
open interest; compute the target functional on the retained set
(plug-in estimator, no error estimate).

\item \textbf{Uniform}: retain each underlying independently with probability $\rho$;
compute the plug-in estimate on the retained set (no design-based correction).

\item \textbf{Stratified}: stratify underlyings into OI quartiles, then sample
proportionally within strata; plug-in estimate.

\item \textbf{Poisson PPS-\Hajek{} (PH)}: independent Poisson sampling with
$\pi_i \propto w_i$ calibrated by water-filling \eqref{eq:waterfill};
\Hajek{} ratio estimator~\eqref{eq:hajek} with linearised design-based
variance~\eqref{eq:hajekvar}.  Sample size is random with mean $n$.

\item \textbf{Fixed-size systematic PPS-\Hajek{}}: systematic
probability-proportional-to-size selection using the same $\pi_i$, giving
a deterministic sample size exactly equal to $n$ and hence a hard
bandwidth cap.  Point estimation uses the same \Hajek{} ratio.
\end{enumerate}

All five designs are calibrated to the same expected budget.  We report
intervals for Poisson PPS only.  This is an implementation scope
decision, not a property of the other designs: uniform Bernoulli and
stratified sampling admit standard design-based variance estimators, and
fixed-size systematic PPS admits one as well, but its variance depends on
second-order inclusion probabilities, which for systematic selection have
no simple closed form and are normally handled by replication methods
such as jackknife or balanced repeated replication.  Poisson sampling is
the one design here whose variance is unbiasedly estimable from
first-order probabilities alone, which is precisely why we adopt it for
the interval construction.  Extending calibrated intervals to the
remaining designs is left to future work; we do not claim they are
unavailable.

\subsection{Retention Rates and Repetitions}

We evaluate five retention rates: $\rho \in \{5\%, 10\%, 20\%, 30\%, 50\%\}$.
Each stochastic method (Uniform, Stratified, Poisson-\Hajek{}, Fixed-size
PPS) is repeated with 200 independent random seeds per trading day
(for Fixed-size PPS the ``seed'' determines the random starting point of
the systematic selection).
We report the mean MAE across seeds and days.
Top-$k$ is deterministic and requires no repetition.

\subsection{Evaluation Metrics}

\paragraph{Point accuracy.}
Mean absolute error (MAE) of the estimated aggregate versus the full-panel
ground truth, averaged across days.
Standard deviations are across-day.

\paragraph{Interval calibration.}
For Poisson-\Hajek{}, we report empirical coverage (fraction of days $\times$
seeds where the true $\theta$ falls within the 95\% interval) and mean
interval width.

\paragraph{Statistical testing.}
Daily mean errors are compared using the paired mean difference (pairing
on trading day, averaging across seeds within each day).
We report the conventional paired $t$-statistic descriptively.
Because daily aggregates are serially dependent, we additionally report
moving-block-bootstrap confidence intervals and $p$-values for
serial-dependence-robust inference~\cite{kunsch1989}, with block length
$\lceil T^{1/3}\rceil$ where $T$ is the number of trading days.
Where we report autocorrelation-robust standard errors we use the
Newey--West form~\cite{neweywest1987}.
With 2~markets $\times$ 4~targets $\times$ 5~retention rates the
pre-specified family comprises 40 comparisons.
Because the raw paired-$t$ $p$-values ignore serial dependence, we
apply multiplicity control to the moving-block-bootstrap $p$-values,
which are the primary inferential statistic for this reason.
Holm's procedure~\cite{holm1979} is valid under arbitrary dependence
between tests and is our default; Benjamini--Hochberg~\cite{benjaminihochberg1995}
is reported as a less conservative cross-check.
Of the 40 comparisons, 31 have $p_{\mathrm{MBB}}<0.05$ unadjusted;
25 survive Holm and 31 survive BH at $\alpha=0.05$.  The six cells on
which the two procedures disagree all have raw $p_{\mathrm{MBB}}$
between $0.004$ and $0.037$ and none of them is a flagship comparison,
so the choice of procedure does not affect any headline claim; we report
Holm throughout.
Every flagship cell with $p_{\mathrm{MBB}}<2\times10^{-4}$ (the
Monte-Carlo resolution) is rejected before and after adjustment.
The two borderline OI-weighted IV cells (US~5\%,
$p_{\mathrm{MBB}}=0.30$; HK~30\%, $p_{\mathrm{MBB}}=0.44$) are not
significant before or after adjustment.
The complete family --- all 40 raw $p_{\mathrm{MBB}}$ values with their
Holm-adjusted $p$-values, BH $q$-values, bootstrap CIs, $t$-statistics
and effect sizes --- is provided as
\texttt{supplementary/multiplicity\_audit.csv}, generated by
\texttt{scripts/gen\_multiplicity\_audit.py}.
Table~\ref{tab:significance} reports the block-bootstrap CI and the
descriptive $t$-statistic.  Where we write $p<10^{-4}$ this refers to
the conventional paired $t$-test, which has sufficient resolution.
The moving-block bootstrap uses 2{,}000 replicates, so its $p$-value has
minimum resolution $1/2{,}000=5\times10^{-4}$; a reported
$p_{\mathrm{MBB}}=0$ means no sign reversal occurred in any of the
2{,}000 replicates, which we state as
$p_{\mathrm{MBB}}<5\times10^{-4}$ or, equivalently, ``the MBB 95\% CI
excludes zero.''  Both tests agree qualitatively in every case we
report.

\paragraph{Stress analysis.}
We partition trading days into terciles by the \emph{absolute
day-over-day change in the cross-sectional mean ATM implied volatility},
$|\Delta\overline{\mathrm{IV}}_t|$: \emph{calm} (bottom third),
\emph{normal} (middle), \emph{turbulent} (top third).  This measures a
market-level shock between sessions, not within-day cross-sectional
spread; we state it explicitly because the two are easily conflated and
carry different economic meaning.  For each regime we report MAE
\emph{together with} empirical coverage and mean interval width, since a
claim that intervals remain trustworthy under stress cannot be supported
by MAE alone.

%% file: sections/results.tex
\section{Results}\label{sec:results}

\subsection{Flagship Result: OI-Weighted IV}

Table~\ref{tab:flagship} and Figure~\ref{fig:mae} present the central result.
\input{tables/generated}

\begin{figure}[t]
\centering
\includegraphics[width=0.8\textwidth]{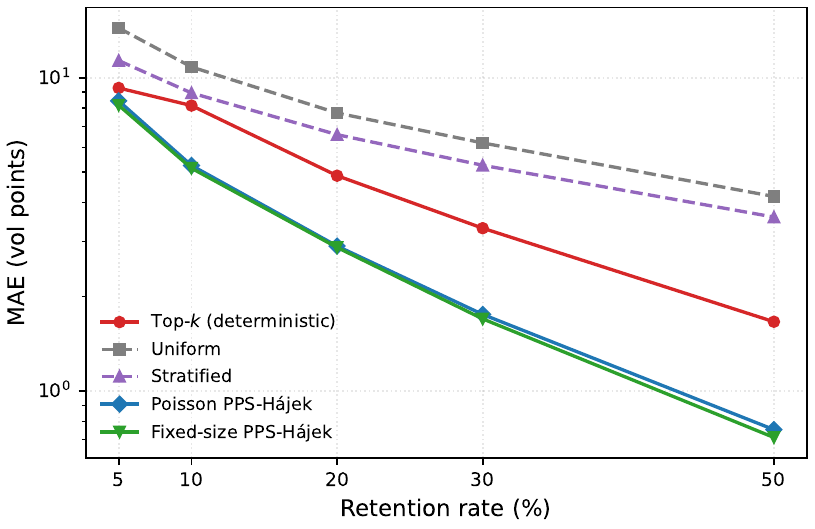}
\caption{US OI-weighted IV: MAE versus retention rate for all five
methods (log scale).  Both probability-proportional designs beat Top-$k$
at every retention rate.  Fixed-size systematic PPS records a nominally
lower MAE than Poisson PPS at all five rates, but that margin is smaller
than its sensitivity to the unit ordering and should not be read as a
ranking (Appendix~\ref{app:ordering}).  We construct intervals from the
Poisson design, whose variance is unbiasedly estimable from first-order
inclusion probabilities alone; see Section~\ref{sec:experiments}.}
\label{fig:mae}
\end{figure}
On the US market both probability-proportional designs beat Top-$k$ at every
retention rate, with the advantage most pronounced at 10\%
(Poisson PPS MAE 5.31 vs.\ 8.15 for Top-$k$, a 35\% reduction;
paired $t=6.72$; moving-block-bootstrap CI excludes zero,
$p_{\mathrm{MBB}}<2\times10^{-4}$; Cohen's $d = 0.75$).
Between the two PPS variants, fixed-size systematic PPS records a
nominally lower MAE than independent Poisson PPS at all five rates (at
10\%, 5.15 vs.\ 5.31).  We deliberately do \emph{not} claim this as a
ranking.  Systematic PPS depends on the order in which units are swept,
and a sensitivity analysis over six orderings
(Appendix~\ref{app:ordering}) shows that the spread induced by the
ordering alone exceeds the Poisson-versus-fixed margin in 9 of 10
market\,$\times$\,rate cells.  The margin is therefore within
implementation noise, and the honest statement is that the two PPS
variants are empirically indistinguishable on accuracy.  The grounds for
choosing between them are structural rather than empirical: fixed-size
PPS gives a hard bandwidth cap, while Poisson PPS retains the simple
unbiased variance estimator for the Horvitz--Thompson total (and a
first-order linearised approximation for the \Hajek{} ratio), which
needs only first-order inclusion probabilities.  We construct all
intervals from the Poisson design for that reason.
At 5\% retention the Top-$k$ gap is not statistically significant
(MAE 8.47 vs.\ 9.28, block-bootstrap $p = 0.33$, $d = 0.13$),
reflecting increased sampling noise at extreme compression.
Interval coverage rises from 89.0\% at 5\% retention to 94.3\% at 50\%.
Every entry lies \emph{below} the nominal 95\% level, i.e.\ the intervals
are mildly \emph{anti-conservative} (under-covering), and the shortfall
grows as retention falls.  This is the expected failure mode of a
first-order linearisation: at $\rho=5\%$ the effective sample is roughly
$\lceil 0.05\times 210\rceil\approx 11$ underlyings, too few for the
normal approximation to the \Hajek{} ratio to be accurate.  We therefore
report these as nominal-95\% approximate intervals with empirical
coverage in the $89$--$94\,\%$ range for the aligned target, and do not
claim guaranteed coverage.
Figure~\ref{fig:coverage} shows coverage across all four targets under the
common OI-proportional design.  After the multiplier-bootstrap correction
only the OI-weighted target approaches the nominal line as retention
rises (89.0--94.3\%); the median PCR intervals remain materially
under-calibrated at every rate (79.4--90.5\%), as do the equal-weight VRP
intervals (78.3--90.4\%), and the IV-dispersion intervals badly so
(61.2--77.2\%), never exceeding $80\%$.  This is the visible cost of design
mismatch that Section~\ref{sec:designaware} repairs.

\begin{figure}[t]
\centering
\includegraphics[width=0.8\textwidth]{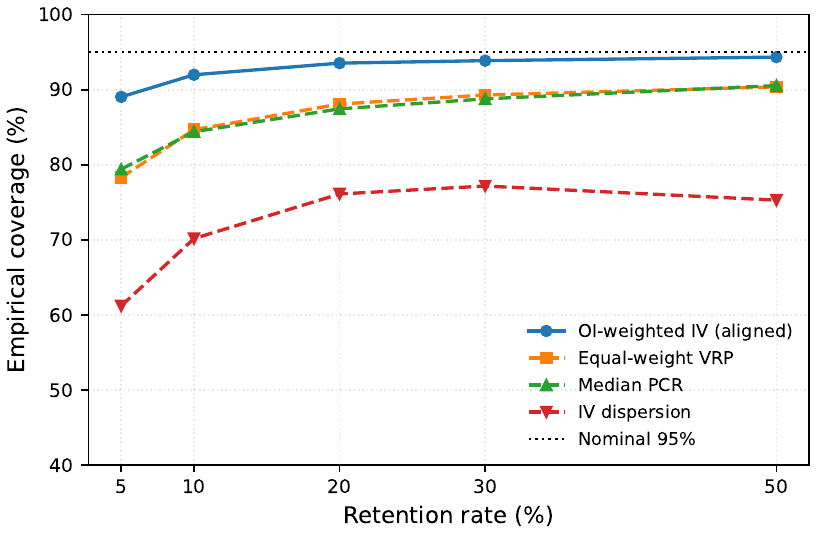}
\caption{Empirical coverage of nominal-95\% intervals (US) versus
retention rate, all under the \emph{same} OI-proportional Poisson
design, after the multiplier-bootstrap correction defined in
Section~\ref{sec:estimators}. Only $T_{\mathrm{oi}}$ is aligned with this
design in the sense of Proposition~\ref{prop:optalloc}; by \eqref{eq:ifs}
the equal-weight mean and the dispersion target have
influence-proportional oracle optima (outcome-dependent, hence
infeasible; uniform is only an outcome-blind proxy), and the median's
oracle optimum is approximately constant in magnitude, so all three are
mis-specified under an OI-proportional design.  Only the aligned target approaches
the nominal line as retention rises; the median and equal-weight
intervals remain materially under-calibrated at every rate, and the
dispersion intervals badly so, never exceeding $80\%$.  Coverage is below
nominal in every cell.}
\label{fig:coverage}
\end{figure}

\paragraph{Hong Kong reversal.}
On HK OI-weighted IV the ordering reverses at low-to-moderate retention:
Top-$k$ attains lower MAE than the sampling designs at 5--30\%
(at 10\%, Top-$k$ 1.33 vs.\ Poisson PPS 2.45; paired-$t$ $p < 10^{-4}$, $d = -1.00$).
The advantage is not uniform, however: by 50\% retention both PPS variants
have overtaken it (Top-$k$ 0.3078, fixed-size PPS 0.2462, Poisson PPS
0.2577), and at 30\% the difference is already statistically
indistinguishable ($p = 0.44$).  The claim is therefore convergence and
eventual reversal with retention, not blanket Top-$k$ superiority in HK.
Section~\ref{sec:concentration} decomposes the effect and shows it is only
partly attributable to HK's higher OI concentration.
Crucially, however, \emph{Top-$k$ still provides no design-based confidence interval}.
The practitioner using Top-$k$ in HK cannot know whether the favorable
MAE persists on any given day without external validation---precisely the
auditability gap our framework addresses.

Table~\ref{tab:feasible} extends the comparison to feasible bounds.
The oracle bounds of Table~\ref{tab:partialid} use the observed IV range
from the full panel---information that is, by definition, unavailable
after Top-$k$ discards data.  The feasible bounds use $[a,b]=[0,500]$
vol points, which can be specified before any sampling or retention
decision and require no discarded data.  This bound is conditionally
worst-case valid under the maintained domain assumption
$x_i\in[0,500]$; in our sample the maximum observed ATM IV is 321.2
(US) and 155.0 (HK), with zero records exceeding 500 out of 16{,}856
(US) and 10{,}773 (HK) core records after de-duplication, so the
maintained assumption is not violated within the observation window.  Under these bounds the Top-$k$
worst-case interval is \IdFeasLo--\IdFeasHi$\times$ wider than the
sampling interval, compared with \IdWidthLo--\IdWidthHi$\times$ under
oracle bounds.  Both types of interval are worst-case valid, whereas the
sampling interval is only nominally 95\%, so the comparison is of
width, not coverage.

\input{tables/feasible_bounds}

\subsection{Design--Target Alignment}\label{sec:alignment}

Table~\ref{tab:alignment} demonstrates that the superiority of
Poisson-\Hajek{} with OI-proportional $\pi_i$ is specific to OI-weighted
targets.
For equal-weight VRP (US, 10\%), stratified sampling achieves the lowest MAE
(4.90), followed by uniform (6.22) and Top-$k$ (7.63); OI-proportional
Poisson-\Hajek{} is worst (4.28).
The pattern is starker for median PCR: at US 10\%, Top-$k$ achieves MAE 0.107
versus PH's 0.262.
These results confirm that inclusion probabilities must match the target
functional---OI-proportional $\pi_i$ over-represents large names that
contribute equally to unweighted aggregates.

\subsection{A Feasible Uniform Correction Consistent with Design--Target Alignment}
\label{sec:designaware-emp}

For equal-weight VRP, Proposition~\ref{prop:optalloc} predicts the
oracle optimum $\pi_i^\star\propto|v_i-T|$, which is outcome-dependent
and infeasible.  Uniform sampling ($\pi_i=\rho$) is not the
theoretical optimum but is a feasible outcome-blind alternative that
removes the irrelevant OI tilt.
We test this directly (script \texttt{design\_aware.py}): we re-run the
equal-weight VRP experiment with uniform Poisson inclusion probabilities
and the \Hajek{} mean estimator (``Unif-PH''), and compare against the
OI-proportional Poisson-\Hajek{} (``OI-PH'') used elsewhere.
Table~\ref{tab:alignment} reports the outcome.

The correction is decisive.
Unif-PH beats OI-PH at \emph{every} retention rate on \emph{both} markets,
in both point accuracy and interval calibration.
At US 10\%, MAE falls from 4.41 to 2.98 while coverage rises from
83.9\% to 91.6\%.
The gap is even larger in HK, where OI is more concentrated:
at 10\%, MAE falls from 3.78 to 2.32 and coverage from
70.0\% to 90.3\%.
The full interval diagnostics for both designs, at all five rates and both
markets, is in Table~\ref{tab:ewcert}.
Across all ten market$\times$rate cells the improvement is significant,
with $|d|\ge1.36$ and no sign reversal in 2{,}000 moving-block-bootstrap
replicates in any cell
(\texttt{scripts/ew\_alignment\_significance.json}).  At US 10\% the
paired statistic is $t=21.3$, $d=2.47$; at HK 10\%, $t=34.5$, $d=3.91$.
This is directionally consistent with the design-alignment framework of
Proposition~\ref{prop:optalloc} and rejects a universal OI-proportional
design: moving away from a design that is mis-aligned with the target's
influence function both lowers error and repairs much of the
under-coverage.  It is not a confirmation of the proposition's optimum,
since uniform is a feasible proxy rather than the (infeasible)
influence-proportional oracle for this target.
Note that Unif-PH is now competitive with the best plug-in baseline
(stratified) on point accuracy while additionally supplying an approximate
interval---the property no plug-in method provides.

\subsection{Cross-Market Comparison}

Table~\ref{tab:crossmarket} compares OI-weighted IV results across markets
at 10\% retention.
The US and HK markets exhibit opposite rankings on point accuracy.
Coverage is reported for the Poisson design (US 91.9\%, HK 90.2\%);
for Top-$k$ the table instead shows the worst-case identification width
of Proposition~\ref{prop:partialid}, which is a different guarantee and
not directly comparable.  Coverage for the remaining probability designs
is not reported here for the implementation reason given in
Section~\ref{sec:experiments}.

\input{tables/crossmarket}

\subsection{Known Total: Horvitz--Thompson versus \Hajek{}}
\label{sec:htknown}

Because Definition~\ref{def:info} places the full auxiliary vector
$w_{1:N}$ in the information set, the population total $W=\sum_i w_i$ is a
\emph{known constant}, and the natural estimator of the OI-weighted mean
is the Horvitz--Thompson form $\widehat T_{\mathrm{HT}}
=W^{-1}\sum_{i\in S}w_ix_i/\pi_i$, which is design-unbiased.  The \Hajek{}
ratio instead re-estimates $W$ from the sample.  Table~\ref{tab:htknown}
reports both.

The comparison is a genuine trade rather than a dominance result.  HT's
intervals are better calibrated, most visibly in Hong Kong at
5\% retention where HT attains 91.5\% coverage against \Hajek{}'s 84.9\%.
The two estimators' mean signed errors are comparable at low retention
($-0.109$ for HT vs.\ $-0.103$ for \Hajek{} in the US at 5\%), so the
calibration advantage comes from the variance side rather than from a
bias reduction.
But HT's mean absolute error is roughly twice \Hajek{}'s throughout
(10.462 vs.\ 5.305 in the US at 10\%), because dividing by the realised
random sample content rather than by a fixed constant removes a large
common factor of sampling variation.  This is the classical
ratio-estimator trade, and it is the reason we retain \Hajek{} for the
headline point-accuracy comparison while noting that an analyst who cares
primarily about calibrated uncertainty rather than point error should
prefer HT.  A calibration or GREG estimator that uses the known $W$ as a
constraint while retaining the variance reduction is the natural next
step and we do not pursue it here.

\input{tables/htknown}

\subsection{What Drives the HK Reversal}\label{sec:concentration}

The Hong Kong reversal is \emph{consistent with, and partly mechanically
driven by}, higher open-interest concentration, but concentration alone
does not account for it.  The exact identity is
\begin{equation}
\widehat T_{\text{Top-}k}-T
=\frac{W_{S^c}}{W}\bigl(\bar x_{w,S}-\bar x_{w,S^c}\bigr),
\label{eq:decomp}
\end{equation}
where $\bar x_{w,S}$ and $\bar x_{w,S^c}$ are the OI-weighted mean
outcomes among retained and omitted names.  Concentration determines the
first factor only; the second --- how different the discarded names are
--- is an empirical property of the panel that concentration does not
constrain.  A panel can omit very little weight and still be badly
mis-estimated if the omitted names are extreme.

Table~\ref{tab:concentration} confirms that HK is markedly more
concentrated on every measure, and Table~\ref{tab:decomp} evaluates both
factors of~\eqref{eq:decomp} directly (the identity holds to machine
precision, residual $<10^{-13}$).  At 10\% retention the omitted OI share
is 55.5\% in the US against 37.7\% in HK, a ratio of only
$1.47\times$; the weighted outcome gap is 14.74 against 3.56, a ratio of
$4.14\times$.  The dominant term is therefore the outcome gap, not
concentration: US names outside the OI top decile have systematically far
more extreme implied volatilities than their HK counterparts.  We
therefore report concentration as a useful \emph{ex ante} screen --- it
bounds the first factor and can be computed before any sampling ---
while noting that it cannot by itself predict the ranking.

Concentration governs point-estimate risk only, not auditability.
However concentrated the panel, a target-blind rule supplies no
design-based interval; under outcome bounds it supplies the worst-case
interval of Proposition~\ref{prop:partialid}, which is
\IdWidthLo--\IdWidthHi$\times$ wider than the sampling interval but valid
in the worst case rather than on average.

\input{tables/decomposition}

\input{tables/concentration}

\subsection{Stress Regime Stability}

At 10\% retention on US OI-weighted IV, Poisson-\Hajek{} MAE is stable
across regimes: 5.52 (calm, 27 days), 5.21 (normal, 26 days), and
5.21 (turbulent, 26 days).
Top-$k$ shows comparable stability (7.89, 8.72, 8.01) but at uniformly
higher error.
Coverage is likewise near-flat across regimes (91.4\%, 91.8\%, 92.4\%), a spread of well under one
percentage point; the intervals therefore neither improve nor deteriorate
materially during volatility spikes. We note this is stability, not
correctness: coverage sits below the nominal level in every regime.

For equal-weight VRP at 10\%, PH MAE is similarly stable across regimes
(4.58, 4.43, 4.25) but consistently higher than stratified sampling
(2.66, 2.72, 2.47), confirming that the design-alignment issue
persists regardless of market conditions.

%% file: tables/generated.tex
\begin{table}[t]
\centering
\caption{OI-weighted implied volatility: mean absolute error against the
full-panel target, by market and retention rate $\rho$. All designs are
calibrated to the same expected budget via~\eqref{eq:waterfill}; realised
sample sizes confirm the match. Coverage is the empirical rate of
nominal-95\% intervals, averaged over trading days.}
\label{tab:flagship}
\small
\begin{tabular}{llrrrrrr}
\toprule
& & \multicolumn{4}{c}{MAE by design} & \multicolumn{2}{c}{PPS interval} \\
\cmidrule(lr){3-6}\cmidrule(lr){7-8}
Mkt & $\rho$ & Top-$k$ & Unif. & Strat. & Poisson PPS & Cov. & $\E|S|$/$|S|$ \\
\midrule
US & 5\% & 9.277 & 14.382 & 11.456 & \textbf{8.469} & 89.3\% & 10.6/10.6 \\
 & 10\% & 8.150 & 10.905 & 8.869 & \textbf{5.305} & 91.9\% & 21.1/21.1 \\
 & 20\% & 4.867 & 7.762 & 6.615 & \textbf{2.920} & 93.3\% & 42.2/42.2 \\
 & 30\% & 3.308 & 6.108 & 5.271 & \textbf{1.760} & 94.0\% & 63.2/63.2 \\
 & 50\% & 1.663 & 4.102 & 3.614 & \textbf{0.755} & 94.5\% & 105.4/105.4 \\
\midrule
HK & 5\% & \textbf{3.303} & 7.686 & 5.901 & 4.035 & 84.9\% & 7.0/7.0 \\
 & 10\% & \textbf{1.324} & 6.030 & 5.035 & 2.423 & 90.2\% & 13.1/13.1 \\
 & 20\% & \textbf{0.663} & 4.130 & 3.618 & 1.158 & 92.7\% & 26.1/26.0 \\
 & 30\% & \textbf{0.597} & 3.176 & 2.719 & 0.654 & 93.3\% & 39.2/39.2 \\
 & 50\% & 0.305 & 2.062 & 1.774 & \textbf{0.258} & 94.5\% & 66.1/66.1 \\
\bottomrule
\end{tabular}
\end{table}

\begin{table}[t]
\centering
\caption{Width of the sharp worst-case identification interval for Top-$k$
(Proposition~\ref{prop:partialid}, with $[a,b]$ the observed IV range from the
full panel --- oracle bounds unavailable after Top-$k$ retention)
versus the width of the design-based interval at the same budget.  See
Table~\ref{tab:feasible} for the wider intervals under feasible bounds.
Target: OI-weighted IV.}
\label{tab:partialid}
\small
\begin{tabular}{llrrr}
\toprule
Mkt & $\rho$ & Top-$k$ worst case & PPS interval & Ratio \\
\midrule
US & 5\% & 135.47 & 37.24 & 3.6$\times$ \\
 & 10\% & 107.11 & 24.59 & 4.4$\times$ \\
 & 20\% & 72.19 & 13.87 & 5.2$\times$ \\
 & 30\% & 50.34 & 8.48 & 5.9$\times$ \\
 & 50\% & 23.26 & 3.66 & 6.4$\times$ \\
\midrule
HK & 5\% & 44.09 & 16.51 & 2.7$\times$ \\
 & 10\% & 31.66 & 10.97 & 2.9$\times$ \\
 & 20\% & 17.75 & 5.40 & 3.3$\times$ \\
 & 30\% & 11.39 & 3.10 & 3.7$\times$ \\
 & 50\% & 5.37 & 1.24 & 4.3$\times$ \\
\bottomrule
\end{tabular}
\end{table}

\begin{table}[t]
\centering
\caption{Design--target alignment at $\rho=10\%$ (US). By~\eqref{eq:ifs}
only the OI-weighted target calls for size-proportional inclusion
probabilities. The equal-weight mean and dispersion targets have
influence-proportional oracle optima that are outcome-dependent and hence
infeasible; uniform is an outcome-blind feasible proxy that removes the
incorrect OI tilt rather than the oracle optimum itself. Only for the median
is the oracle optimum itself approximately constant in magnitude. MAE of the best
design per row in bold.}
\label{tab:alignment}
\small
\begin{tabular}{lrrrr}
\toprule
Target & PPS & Uniform & Stratified & Fixed-PPS \\
\midrule
OI-weighted IV & 5.305 & 10.905 & 8.869 & \textbf{5.149} \\
Equal-weight VRP & 4.415 & 2.984 & \textbf{2.617} & 4.357 \\
Median PCR & 0.254 & 0.116 & \textbf{0.111} & 0.268 \\
IV dispersion & 6.407 & 6.222 & \textbf{4.903} & 6.324 \\
\bottomrule
\end{tabular}
\end{table}

\begin{table}[t]
\centering
\caption{Stress regimes, defined by terciles of the absolute day-over-day
change in cross-sectional mean ATM IV. We report coverage and interval
width alongside MAE, since stability of the intervals cannot be inferred
from MAE alone. Target: OI-weighted IV.}
\label{tab:stress}
\small
\begin{tabular}{lllrrr}
\toprule
Mkt & $\rho$ & Regime & Coverage & Width & MAE \\
\midrule
US & 10\% & calm & 91.4\% & 25.38 & 5.524 \\
 &  & normal & 91.8\% & 24.05 & 5.208 \\
 &  & turbulent & 92.4\% & 24.42 & 5.208 \\
\midrule
 & 30\% & calm & 94.3\% & 8.75 & 1.794 \\
 &  & normal & 94.5\% & 8.14 & 1.688 \\
 &  & turbulent & 93.3\% & 8.61 & 1.806 \\
\midrule
HK & 10\% & calm & 90.9\% & 11.47 & 2.538 \\
 &  & normal & 90.1\% & 10.56 & 2.342 \\
 &  & turbulent & 89.7\% & 10.95 & 2.404 \\
\midrule
 & 30\% & calm & 93.0\% & 3.23 & 0.692 \\
 &  & normal & 93.6\% & 2.94 & 0.618 \\
 &  & turbulent & 93.4\% & 3.16 & 0.657 \\
\bottomrule
\end{tabular}
\end{table}

\begin{table}[t]
\centering
\caption{Paired comparison of Top-$k$ minus PPS daily MAE. The bootstrap
is a moving-block bootstrap (block length $\lceil T^{1/3}\rceil$) to
respect serial dependence across trading days. Positive differences
favour PPS.}
\label{tab:significance}
\small
\begin{tabular}{llrrrl}
\toprule
Mkt & $\rho$ & Mean diff. & $t$ & Cohen's $d$ & Block-bootstrap 95\% CI \\
\midrule
US & 5\% & 0.809 & 1.17 & 0.13 & $[-0.872,\,2.684]$ \\
 & 10\% & 2.845 & 6.68 & 0.75 & $[1.822,\,3.957]$ \\
 & 30\% & 1.549 & 8.09 & 0.90 & $[1.030,\,2.131]$ \\
\midrule
HK & 5\% & -0.733 & -3.18 & -0.35 & $[-1.239,\,-0.101]$ \\
 & 10\% & -1.099 & -8.92 & -0.98 & $[-1.362,\,-0.812]$ \\
 & 30\% & -0.057 & -1.20 & -0.13 & $[-0.176,\,0.085]$ \\
\bottomrule
\end{tabular}
\end{table}

%% file: tables/feasible_bounds.tex
\begin{table}[t]
\centering
\caption{Worst-case identification interval widths for Top-$k$ under
oracle bounds (the observed IV range from the full panel, which uses
discarded data and is therefore not available in production) and under
feasible bounds ($[0,500]$ vol points, which can be specified before
sampling and require no discarded data). The design-based interval
width from Poisson PPS is shown for comparison. Oracle and feasible
widths are worst-case valid; the PPS width targets 95\% nominal
coverage with empirical coverage below nominal (Section~\ref{sec:results}).
The two types of interval are therefore not interchangeable.}
\label{tab:feasible}
\small
\begin{tabular}{llrrrrr}
\toprule
Mkt & $\rho$ & Oracle & Feasible & PPS & Oracle/$\,$PPS & Feasible/$\,$PPS \\
\midrule
US & 5\% & 135.5 & 351.4 & 37.2 & 3.6$\times$ & 9.4$\times$ \\
 & 10\% & 107.1 & 277.3 & 24.5 & 4.4$\times$ & 11.3$\times$ \\
 & 20\% & 72.2 & 186.2 & 13.9 & 5.2$\times$ & 13.4$\times$ \\
 & 30\% & 50.3 & 129.3 & 8.5 & 5.9$\times$ & 15.2$\times$ \\
 & 50\% & 23.3 & 59.6 & 3.7 & 6.4$\times$ & 16.3$\times$ \\
\midrule
HK & 5\% & 44.0 & 262.7 & 16.5 & 2.7$\times$ & 15.9$\times$ \\
 & 10\% & 31.6 & 188.4 & 11.0 & 2.9$\times$ & 17.1$\times$ \\
 & 20\% & 17.7 & 105.3 & 5.4 & 3.3$\times$ & 19.5$\times$ \\
 & 30\% & 11.3 & 67.7 & 3.1 & 3.7$\times$ & 21.9$\times$ \\
 & 50\% & 5.3 & 31.9 & 1.2 & 4.3$\times$ & 25.8$\times$ \\
\bottomrule
\end{tabular}
\end{table}

%% file: tables/crossmarket.tex
\begin{table}[t]
\centering
\caption{OI-weighted IV at 10\% retention. Every probability design admits a design-based interval; only the target-blind Top-$k$ rule does not, and is therefore shown with its worst-case identification width from Proposition~\ref{prop:partialid} instead.}
\label{tab:crossmarket}
\small
\begin{tabular}{llrrrr}
\toprule
Market & Method & MAE & SD & Cov.\% & CI / ID Width \\
\midrule
\multirow{5}{*}{US} & Top-$k$ & 8.15 & 4.17 & n/a & 107.11$^{\dagger}$ \\
 & Uniform & 10.91 & 1.61 & -- & -- \\
 & Stratified & 8.87 & 1.43 & -- & -- \\
 & Fixed-PPS & \textbf{5.15} & 0.80 & -- & -- \\
 & PPS-\Hajek{} & 5.31 & 0.87 & 91.9 & 24.59 \\
\midrule
\multirow{5}{*}{HK} & Top-$k$ & \textbf{1.32} & 1.16 & n/a & 31.66$^{\dagger}$ \\
 & Uniform & 6.03 & 1.70 & -- & -- \\
 & Stratified & 5.03 & 1.36 & -- & -- \\
 & Fixed-PPS & 2.41 & 0.47 & -- & -- \\
 & PPS-\Hajek{} & 2.42 & 0.49 & 90.2 & 10.97 \\
\bottomrule
\end{tabular}
\par\smallskip\footnotesize $^{\dagger}$Oracle worst-case identification width using the full-panel observed IV range; unavailable to a production consumer after Top-$k$ retention, and not a confidence interval. Feasible bounds ($[0,500]$) give substantially wider intervals; see Table~\ref{tab:feasible}.
\end{table}

%% file: tables/htknown.tex
\begin{table}[t]
\centering
\caption{Horvitz--Thompson with known population total $W$ versus the \Hajek{} ratio, both under the same OI-proportional Poisson design. The \Hajek{} columns are drawn from Table~\ref{tab:flagship} to ensure bit-for-bit consistency. Because $w_{1:N}$ is in the information set, $W$ is known and HT is design-unbiased; the \Hajek{} ratio re-estimates $W$ and trades bias for variance.}
\label{tab:htknown}
\small
\begin{tabular}{llrrrrrr}
\toprule
& & \multicolumn{3}{c}{HT (known $W$)} & \multicolumn{3}{c}{\Hajek{} ratio} \\
\cmidrule(lr){3-5}\cmidrule(lr){6-8}
Mkt & $\rho$ & MAE & Cov.\% & Width & MAE & Cov.\% & Width \\
\midrule
US & 5\% & 16.484 & 91.4 & 80.05 & 8.469 & 89.3 & 37.24 \\
 & 10\% & 10.462 & 93.3 & 51.41 & 5.305 & 91.9 & 24.59 \\
 & 20\% & 5.871 & 93.8 & 28.55 & 2.920 & 93.3 & 13.87 \\
 & 30\% & 3.606 & 94.1 & 17.56 & 1.760 & 94.0 & 8.48 \\
 & 50\% & 1.522 & 94.5 & 7.49 & 0.755 & 94.5 & 3.66 \\
\midrule
HK & 5\% & 10.511 & 91.2 & 51.41 & 4.035 & 84.9 & 16.51 \\
 & 10\% & 6.328 & 91.9 & 30.28 & 2.423 & 90.2 & 10.97 \\
 & 20\% & 3.067 & 93.1 & 14.91 & 1.158 & 92.7 & 5.40 \\
 & 30\% & 1.709 & 93.6 & 8.38 & 0.654 & 93.3 & 3.10 \\
 & 50\% & 0.658 & 94.1 & 3.18 & 0.258 & 94.5 & 1.24 \\
\bottomrule
\end{tabular}
\end{table}

%% file: tables/decomposition.tex
\begin{table}[t]
\centering
\caption{Exact decomposition of Top-$k$ error at 10\% retention, $\widehat T_{\text{Top-}k}-T=(W_{S^c}/W)(\bar x_{w,S}-\bar x_{w,S^c})$. Concentration determines only the first factor; the realised error also depends on the retained-minus-omitted outcome gap, which concentration does not constrain.}
\label{tab:decomp}
\small
\begin{tabular}{lrrrr}
\toprule
Mkt & Omitted OI share & $|\bar x_{w,S}-\bar x_{w,S^c}|$ & Product & Top-$k$ MAE \\
\midrule
US & 55.5\% & 14.74 & 8.151 & 8.151 \\
HK & 37.7\% & 3.56 & 1.333 & 1.333 \\
\bottomrule
\end{tabular}
\end{table}

%% file: tables/concentration.tex
\begin{table}[t]
\caption{Open-interest concentration. Gini and HHI are computed per session and averaged; Top-10\% OI is the share of total open interest held by the largest decile of names.}
\label{tab:concentration}
\centering\small
\begin{tabular}{@{}lrrrrr@{}}
\toprule
Mkt & Sessions & Gini & HHI & Top-10\% OI & Panel \\
\midrule
US & 80 & 0.590\tiny{$\pm$0.028} & 0.016\tiny{$\pm$0.007} & 44.5\% & 210.7 \\
HK & 82 & 0.733\tiny{$\pm$0.005} & 0.046\tiny{$\pm$0.002} & 62.3\% & 131.4 \\
\bottomrule
\end{tabular}
\end{table}

%% file: sections/related.tex
\section{Related Work}\label{sec:related}

\paragraph{Design-based inference for finite populations.}
The estimators we use are classical.  Horvitz--Thompson
estimation~\cite{horvitz1952}, the \Hajek{} ratio~\cite{hajek1971},
optimal allocation~\cite{neyman1934} and unequal-probability designs are
standard material~\cite{cochran1977,sarndal1992,lohr2019}.  Our
contribution is not a new estimator but the observation that a data
\emph{engineering} decision --- which rows a pipeline keeps under a
bandwidth cap --- is exactly a sampling-design decision, and that the
prevailing engineering default (keep the largest rows) forfeits the
inference that the sampling literature takes for granted.  We use known
auxiliary totals in the sense of calibration~\cite{deville1998} when
comparing against the Horvitz--Thompson estimator with known $W$ in
Section~\ref{sec:htknown}, though we do not pursue a full GREG
construction.

\paragraph{Fixed-size unequal-probability selection.}
Our fixed-size design uses systematic PPS selection; the difficulty of
its second-order inclusion probabilities, and hence of closed-form
variance estimation, is well known, and Chao's procedure~\cite{chao1982}
is one classical response.  This is precisely why we take the Poisson
design as the vehicle for interval construction: it is the design whose
variance is unbiasedly estimable from first-order probabilities alone.

\paragraph{Variance estimation for non-smooth functionals.}
Medians and quantiles are not smooth functionals of the sampling weights,
so their design variance requires either linearisation via the
density-weighted influence function~\cite{francisco1991quantile,binder1991median}
or resampling.  We use a multiplier bootstrap on the design weights,
matched to the Poisson design by a $\sqrt{1-\pi_i}$ factor so that
certainty units receive no perturbation; without this factor the
multiplier scheme over-disperses.  The statistical testing
protocol---moving-block bootstrap, Newey--West robust standard errors,
and Benjamini--Hochberg multiplicity control---is described in
Section~\ref{sec:experiments}.

\paragraph{Approximate query processing and sketching.}
The database community has long studied answering aggregate queries from
a reduced representation, through sampling-based online
aggregation and systems such as BlinkDB~\cite{agarwal2013blinkdb}, through
synopsis and sketch structures~\cite{cormode2012synopses}, and in the
broader AQP literature~\cite{chaudhuri2017aqp}.  Those systems typically
control error relative to a stored sample of a \emph{known} design, which
is exactly the property a Top-$k$ retention rule destroys.  Our setting
differs in that the reduction is imposed upstream by a bandwidth or
licensing constraint rather than chosen by the query engine, and the
population is a daily cross-section whose weights are known in advance
--- which is what makes a design-based treatment feasible at all.

\paragraph{Optimal subsampling and coresets.}
Statistical subsampling for large-scale estimation derives sampling
probabilities proportional to an influence or leverage
quantity~\cite{ma2015statistical,wang2018optimal}, closely paralleled by
leverage-score sampling in randomized numerical linear
algebra~\cite{drineas2011faster} and by coreset
constructions~\cite{feldman2020coresets}.  Proposition~\ref{prop:optalloc}
is the same principle specialised to finite-population aggregates, and we
use it diagnostically: it predicts which of our four targets a single
OI-proportional sampler will serve well, and that prediction is what
Section~\ref{sec:alignment} tests.  Unlike the coreset literature we do
not seek a worst-case multiplicative approximation guarantee; we seek an
interval whose width is computable from the retained sample.

\paragraph{Partial identification.}
When the retained sample does not determine the target, the appropriate
object is an identified set rather than a point estimate.  Our
Proposition~\ref{prop:partialid} is an elementary instance of the
worst-case bounding argument familiar from the partial-identification
literature~\cite{manski1990}: with outcome bounds $[a,b]$ the omitted mass
can be placed at either extreme.  The practical caveat we emphasise in
Section~\ref{sec:results} is that such bounds must not be read off the
discarded data, or the resulting interval is an oracle quantity.

\paragraph{Model risk and supervisory context.}
Supervisory guidance on model risk management~\cite{sr1107} requires institutions to identify and quantify the
uncertainty introduced by their data and modelling choices.  A downsampling
step that reports an aggregate with no accompanying uncertainty statement
is precisely an unquantified source of such risk, which is the practical
motivation for this work.  In derivatives analytics the aggregates
themselves --- implied-volatility levels and variance-risk-premium
proxies --- follow standard constructions~\cite{carrwu2009}.

\paragraph{Relation to KV-cache eviction.}
The impossibility argument we build on was developed for deterministic
KV-cache eviction in large language models~\cite{kvcert2026}.  The present
paper transfers that argument to a finite-population setting where the
weights are known and the target is a cross-sectional aggregate, which
changes both the estimator and the available guarantees: here the omitted
mass is measurable ex ante, so non-identification can be sharpened into a
computable worst-case interval rather than left as an impossibility.

%% file: sections/discussion.tex
\section{Discussion}\label{sec:discussion}

\paragraph{Auditability vs.\ point accuracy.}
The Hong Kong reversal (Section~\ref{sec:results}) might suggest that
Top-$k$ is ``good enough'' when OI concentration is high.
We argue this misses the point.
The value of design-based confidence intervals is not contingent on winning every
point-accuracy comparison---it lies in \emph{auditability}.
A risk manager receiving a Top-$k$ aggregate of HK OI-weighted IV
(MAE 1.33 at 10\%) has no way to verify, from the retained data alone,
that the error is small on \emph{that particular day}.
Under the ex-ante feasible bound $[0,500]$ (which requires no
discarded data), the worst-case identification width from
Proposition~\ref{prop:partialid} is 188.4;
the narrower width of 31.6 reported in Table~\ref{tab:partialid}
is a post-hoc oracle benchmark based on the full-panel observed range
and is unavailable to the production consumer.
The Poisson-\Hajek{} estimate (MAE 2.45) instead comes with an
interval of mean width 10.98 and 90.4\% empirical coverage against a
nominal 95\%.  The comparison is therefore a trade of a wide but
worst-case-valid range against a narrow but only approximately calibrated
one, and we do not claim the latter dominates at equal reliability.

\paragraph{Design-aware pipelines.}
Proposition~\ref{prop:optalloc} and its empirical validation
(Table~\ref{tab:alignment}) turn the alignment issue from a caveat into
an actionable recipe: choose $\pi_i \propto |\mathrm{IF}_i|$ for the target
at hand.
A pipeline computing multiple aggregates from one panel still faces a
multi-objective design problem, since no single $\bm{\pi}$ is optimal for
all targets simultaneously.
Viable resolutions include maintaining parallel subsamples with
target-specific $\pi_i$ schedules, or a compromise design; but our results
show the naive default of a single OI-proportional subsample can nearly
\emph{double} the error of unweighted targets.
At minimum, a pipeline should report which target each subsample is
tuned for.

\paragraph{Implications for model risk management.}
Financial regulators increasingly require quantification of model
risk~\cite{sr1107}.
Aggregation pipelines that silently discard data introduce a form of
model risk that is currently unreported.
Design-based confidence intervals make this risk explicit and auditable:
an interval width of 24.56 on a US OI-weighted IV aggregate at
10\% retention (Table~\ref{tab:flagship}) tells a risk officer how much
uncertainty the downsampling step contributes, subject to the
approximation caveat above.

\paragraph{Relation to KV-cache eviction intervals.}
Our work builds on the impossibility result of~\cite{kvcert2026}, which
proved that deterministic KV-cache eviction in large language models
yields non-identifiable approximation error.
The two settings share a close structural parallel: in both, a finite-population
aggregate (attention output or market statistic) is approximated from
a subsample, and target-blind deterministic selection forecloses design-based
probabilistic error estimation without additional assumptions.
Our contributions are domain-specific: the \Hajek{} ratio estimator for
weighted means, multi-market empirical validation with real options data,
the design-alignment analysis, and the stress-regime robustness tests
are all absent from the machine-learning setting.

\paragraph{Limitations.}
We list six limitations of the present study and, where possible,
quantify the direction of their likely effect.

\emph{(i) Sample length.}
Our sample covers 80 US and 82 HK trading sessions, roughly four
calendar months each. Power is adequate for the prespecified primary
comparison --- Top-$k$ versus Poisson PPS at US $\rho=10\%$ has
$t=6.72$ and $p_{\mathrm{MBB}}<2\times10^{-4}$ --- but it is not uniform
across the design grid. Several secondary retention-rate comparisons are
not statistically significant, including US $\rho=5\%$ ($t=1.20$,
$p_{\mathrm{MBB}}=0.30$) and HK $\rho=30\%$ ($t=-1.15$,
$p_{\mathrm{MBB}}=0.44$); 6 of the 10 market\,$\times$\,rate cells reach
$|t|>4$. A longer backtest covering a full market cycle (a full year
including at least one macro regime change) would both strengthen
external validity and resolve the underpowered cells.  Extending to
12+ months covering at least one VIX regime change (e.g.\ a
low-volatility or crisis episode) is immediate future work.

\emph{(ii) Single macro regime.}
Our observation window (March--July 2026) coincides with a broad
drawdown in the memory-and-storage sector, so much of our variability
comes from a single macro shock. The stress terciles we report are
computed \emph{within} this window (by day-over-day $|\Delta\bar{\mathrm{IV}}|$),
so they measure within-regime robustness rather than cross-regime
transfer. Intervals would need re-validation in bull, low-volatility,
and crisis regimes.

\emph{(iii) End-of-day snapshots.}
All panels are 16:00 local-time snapshots; intraday panels with
time-varying composition (options that trade only during liquid
hours) may exhibit different Top-$k$ vs.\ Poisson-\Hajek{} rankings.
The non-identifiability result (Prop.~\ref{prop:unident}) is
unaffected, but coverage calibration may differ.

\emph{(iv) Systematic under-coverage of the linearised variance.}
The design-based variance in \eqref{eq:hajekvar} is
unbiased for the Horvitz--Thompson total but only \emph{approximately}
unbiased for the \Hajek{} ratio via linearisation.  Empirically the
intervals under-cover at \emph{every} retention rate we test: on
OI-weighted IV, coverage runs $89.0\%\to94.3\%$ (US) and
$84.9\%\to94.3\%$ (HK) as $\rho$ goes from $5\%$ to $50\%$, always below
the nominal $95\%$.  The shortfall is largest where the effective sample
is smallest ($\lceil 0.05\times 210 \rceil \approx 11$ US,
$\approx 7$ HK).  We therefore describe the construction as an
approximate design-based interval rather than a guaranteed one.

The shortfall is markedly worse for the non-smooth targets.  For the
median and for dispersion the design variance is obtained by a multiplier
bootstrap on the design weights rather than by linearisation, and the
multiplier must be matched to the design: with $\xi_i\sim\mathrm{Exp}(1)$
the naive weight $\xi_i/\pi_i$ has variance $\pi_i^{-2}\psi_i^2$ and so
omits the $(1-\pi_i)$ factor of the Poisson design variance, which in
particular assigns non-zero variance to a certainty unit with $\pi_i=1$.
We therefore perturb by $1+\sqrt{1-\pi_i}(\xi_i-1)$, which reproduces
$(1-\pi_i)\pi_i^{-2}\psi_i^2$ and leaves certainty units unperturbed.
Correcting this narrows the intervals and consequently \emph{lowers}
measured coverage --- for US dispersion at $30\%$ retention from
$87.0\%$ to 77.2\%, and for the US median at the same rate from
$91.1\%$ to 88.4\%.  We report the corrected, lower numbers.  The
practical implication is that our dispersion intervals are not usable at
low retention as they stand, and we say so rather than retaining an
over-dispersed bootstrap that would have flattered the coverage table.
Obtaining genuine coverage guarantees --- for instance via
Bernstein-type bounds for bounded outcomes under Horvitz--Thompson
sampling, or self-normalised concentration inequalities --- is the
natural next step and is not attempted here.

\emph{(v) Excluded market.}
China A-share options are excluded due to the small panel (mean 6.9
records/day, ETF-only, no VRP or PCR data).  We apply a uniform
exclusion criterion: any market with a mean daily panel size below
20~underlyings is dropped, since subsampling experiments at $\rho=5\%$
would yield fewer than one selected unit on average.  Our framework
applies unchanged when the panel grows.

\emph{(vi) Design-oracle assumption in extended validation.}
The influence-proportional design tested in
Appendix~\ref{app:design-extended} uses full-panel pilot information
to compute $|\mathrm{IF}_i|$. In production one would use a lagged
or sub-sampled pilot; the gap between oracle and feasible pilot is
not quantified here and is a natural direction for follow-up work.

%% file: sections/conclusion.tex
\section{Conclusion}

We have shown that target-blind Top-$k$ downsampling of
cross-sectional options panels does not identify the OI-weighted
aggregate from retained outcomes alone in the absence of valid outcome
bounds. Under such bounds it admits a sharp worst-case identification
interval, which is substantially wider than the approximate interval
produced by probability sampling in our data.
Replacing Top-$k$ with Poisson sampling and \Hajek{} estimation yields
approximate design-based confidence intervals while simultaneously reducing MAE by up
to 36\% on matched OI-weighted targets.
The auditability mechanism transfers across both markets we study, but
the point-accuracy ranking does not: it depends materially on
open-interest concentration and on design--target alignment, and it
reverses on Hong Kong data.

Equally important, our experiments expose a design-alignment principle:
inclusion probabilities must match the target functional.
OI-proportional sampling, directionally aligned with the OI-weighted
target (though the first-order oracle optimum additionally depends on
$|x_i-T|$), is demonstrably suboptimal for equal-weight and median
statistics.
We advocate a \emph{design-aware} paradigm in which production pipelines
explicitly select or document their sampling design relative to each
downstream aggregate.

Approximate design-based intervals can make downsampling uncertainty more
visible and auditable when their empirical calibration is adequate---a
property of increasing importance as financial regulators demand
transparency in data processing pipelines.  Our results also show that
such calibration cannot be assumed across targets or retention rates:
the dispersion target never reaches nominal coverage in our sample and
falls to $42.4\%$ in HK at the most aggressive compression, so the intervals
are a diagnostic to be validated per target, not a guarantee that comes
free with the design.

%% file: sections/appendix_proofs.tex
\section{Full Proofs}\label{app:proofs}

This appendix supplies rigorous versions of the two central theoretical
statements sketched in the body: non-identification of the target under
\emph{target-blind} deterministic selection
(Proposition~\ref{prop:unident}) and the optimal-allocation result for
Poisson designs (Proposition~\ref{prop:optalloc}).  The scope
restriction to target-blind rules is essential and is discussed in
Remark~\ref{rem:blind-needed}.

\subsection{Proof of Proposition~\ref{prop:unident}}
\label{app:proof-nonident}

Throughout this subsection the population size $N$, the auxiliary vector
$w_{1:N}$ and the retained values $x_S$ are held \emph{fixed}: they are
exactly the information set of Definition~\ref{def:info}.  This is the
essential discipline that the argument must respect.  An earlier version
of this appendix instead constructed populations by \emph{appending}
records, which changes $N$ and $w_{1:N}$ and therefore proves a different
(and weaker) statement.  We also emphasise that the result is restricted
to target-blind rules in the sense of Definition~\ref{def:blind};
Remark~\ref{rem:blind-needed} exhibits a deterministic rule for which the
conclusion is false, so no unrestricted version of this proposition can
hold.

\begin{proof}[Full proof]
Let $D$ be target-blind, so that $S=D(w_{1:N})$ with $|S|=n<N$, and fix
$N$, $w_{1:N}$ (all $w_i>0$) and the retained values $x_S$.  Let
$S^c=\{1,\dots,N\}\setminus S$, which is non-empty since $n<N$.

\emph{Step 1: the retained data are invariant to $x_{S^c}$.}
Because $D$ depends on $w_{1:N}$ only, replacing $x_{S^c}$ by any other
vector $\tilde x_{S^c}\in\R^{N-n}$ leaves $S$ unchanged, hence leaves the
observed pair $(x_S,w_{1:N})$ unchanged.  Two populations that differ only
in $x_{S^c}$ are therefore \emph{observationally equivalent}: no function
of the retained data can distinguish them.  Note that this step is where
target-blindness is used, and it is exactly what fails for the
extremum-retaining rule of Remark~\ref{rem:blind-needed}.

\emph{Step 2: the target is unbounded along this equivalence class.}
Write $W=\sum_{i=1}^N w_i$ and fix $j\in S^c$.  As a function of $x_j$
alone,
\[
  T_{\mathrm{oi}}(x_j)
  =\frac{\sum_{i\in S}w_ix_i+\sum_{i\in S^c\setminus\{j\}}w_ix_i+w_jx_j}{W},
\]
which is affine in $x_j$ with slope $w_j/W>0$.  Since $x_{S^c}$ is
unrestricted in $\R^{N-n}$ by hypothesis, letting $x_j\to\pm\infty$ makes
$T_{\mathrm{oi}}\to\pm\infty$ while the observed data stay fixed by
Step~1.  Hence the identified set
\[
  \Theta(x_S,w_{1:N})
  =\bigl\{T_{\mathrm{oi}}(x_S,\tilde x_{S^c},w_{1:N}):
      \tilde x_{S^c}\in\R^{N-n}\bigr\}
\]
is an unbounded interval; by continuity in $x_j$ it is all of $\R$.

\emph{Step 3: consequence for interval-valued reports.}
Any procedure that reports a set $\mathcal{C}(x_S,w_{1:N})\subseteq\R$
computed from the retained data must, to be valid uniformly over
populations consistent with those data, satisfy
$\Theta(x_S,w_{1:N})\subseteq\mathcal{C}(x_S,w_{1:N})$.  By Step~2 the
left-hand side is $\R$, so $\mathcal{C}$ cannot have finite width.  This
is a statement about worst-case validity over an equivalence class, not
about any particular sampling distribution.
\end{proof}

\begin{remark}[Which targets inherit the conclusion, and which do not]
\label{rem:target-scope}
The argument above transfers verbatim to any functional that is
\emph{unbounded and strictly monotone} in at least one unobserved
coordinate at fixed $N$:
\begin{itemize}
\item \textbf{OI-weighted IV} and \textbf{equal-weight VRP}: affine in
each $x_i$ with non-zero coefficient ($w_i/W$ and $1/N$ respectively),
so Step~2 applies directly.
\item \textbf{Cross-sectional dispersion}: for $N-n\ge 1$ the sample
standard deviation tends to $+\infty$ as $|x_j|\to\infty$, so the
identified set is unbounded above (it is bounded below by $0$, hence the
identified set is a half-line rather than all of $\R$).
\item \textbf{Median PCR}: the conclusion \emph{does not hold in
general}, and we state it only under an additional condition.  The
median is defined in the body as an unweighted median over
underlyings (Definition~\eqref{eq:targets2}), so the relevant
quantity is the \emph{count} of unobserved records, not their OI
weight.

A single unobserved coordinate cannot move the median arbitrarily.
Concretely, with $N=100$ and $n=99$ the median lies between two
adjacent retained order statistics regardless of the missing value.
The correct condition is count-based: let $m=N-n=|S^c|$ be the
number of missing records.
For odd $N$, the median is the $\lceil N/2\rceil$-th order statistic;
for even $N$, the median under the standard convention is the average
of the $(N/2)$-th and $(N/2+1)$-th order statistics.
In both cases the identified set is unbounded when
$m\ge\lceil N/2\rceil$: placing the $m$ unobserved values suitably
(all equal to $v$ for odd $N$, or spanning $v$ for even $N$) can
drive the median to any target, so the identified set is all of $\R$.

The even boundary case is genuine: with $N=4$, $m=2=N/2$ and retained
values $\{0,10\}$, the two missing values can be chosen to make the
sample median arbitrarily large or small, so the identified set is
already unbounded at $m=N/2=\lceil N/2\rceil$ (since $N/2=\lceil N/2\rceil$
for even $N$).  This is why the condition must be $m\ge\lceil N/2\rceil$
rather than $m>\lfloor N/2\rfloor$.

When $m<\lceil N/2\rceil$ the identified set is a bounded interval
$[\ell,u]$ whose endpoints are determined by the corresponding
retained order statistics.  The appropriate object is then a
partial-identification interval rather than an impossibility claim.
The earlier ``append many replicates to shift the median freely''
argument is invalid because it varies $N$.
\end{itemize}
Under bounds $x_i\in[a,b]$ all four targets admit finite sharp intervals;
Proposition~\ref{prop:partialid} gives the closed form for the
OI-weighted case, which is the one we evaluate empirically.
\end{remark}

\begin{remark}
The impossibility is \emph{worst-case} over populations consistent with
the retained data.  Placing a prior on $x_{S^c}$ yields a Bayesian
credible interval, but it is not design-based and inherits the prior's
sensitivity.  Probability sampling avoids this by supplying coverage over
the sampling law of $S$ itself rather than over an assumed prior on the
population --- at the price, documented in Section~\ref{sec:results},
that the resulting intervals are only first-order approximations and
under-cover in finite samples.
\end{remark}

\subsection{Proof of Proposition~\ref{prop:optalloc}}
\label{app:proof-optalloc}

We give the full Lagrangian derivation, the KKT conditions, and the
treatment of the boundary case $\pi_i = 1$.

\begin{proof}[Full proof]
Under Poisson sampling with independent inclusion indicators
$I_i \sim \mathrm{Bernoulli}(\pi_i)$, and given a target functional
$T$ estimated by its \Hajek/plug-in form
$\hat{T} = T(\hat{P})$ where $\hat{P}$ is the design-weighted empirical
measure, a first-order expansion \citep{vaart1998, hampel1986} gives
\[
  \hat{T} - T \;=\; \sum_{i=1}^{N} \frac{I_i - \pi_i}{\pi_i}\,\mathrm{IF}_i
       \;+\; o_p(1),
  \qquad \mathrm{IF}_i \equiv \mathrm{IF}(i; T, P).
\]
Because the $I_i$ are independent with $\mathrm{Var}(I_i) = \pi_i(1-\pi_i)$,
the leading-order design variance is
\begin{equation}\label{eq:linearvar}
  V(\bm\pi) \;=\; \sum_{i=1}^{N} \frac{1-\pi_i}{\pi_i}\,\mathrm{IF}_i^{\,2}.
\end{equation}

We minimize $V$ subject to the expected-sample-size budget
$\sum_i \pi_i = n$ and box constraints
$\pi_i \in [\varepsilon,1]$, where $\varepsilon\in(0,1]$ is an
inclusion floor with $N\varepsilon\le n$.  The floor keeps every
inclusion probability strictly positive---required both for
auditability and for the degenerate case $\mathrm{IF}_i=0$, which the
unconstrained solution $\pi_i^\star=\nu^{-1/2}|\mathrm{IF}_i|$
would send to $0$, violating $\pi_i>0$.
This is a convex program in $\bm\pi$: on $[\varepsilon,1]$ each term
$(1-\pi_i)\pi_i^{-1}\mathrm{IF}_i^2 = (\pi_i^{-1} - 1)\mathrm{IF}_i^2$
is convex in $\pi_i$ (its second derivative is
$2\mathrm{IF}_i^2 / \pi_i^3 \ge 0$), and the constraints are affine.
Standard convex optimization \citep{boyd2004} guarantees the
KKT conditions are necessary and sufficient.

\emph{Lagrangian.}
Write $a_i \equiv |\mathrm{IF}_i|$.  Both box constraints carry
multipliers:
\[
  \mathcal{L}(\bm\pi,\nu,\bm\alpha,\bm\beta)
  \;=\; \sum_{i=1}^N \frac{1-\pi_i}{\pi_i}\,a_i^{2}
     \;+\; \nu\!\left(\sum_i \pi_i - n\right)
     \;+\; \sum_i \alpha_i\,(\varepsilon - \pi_i)
     \;+\; \sum_i \beta_i\,(\pi_i - 1),
\]
with $\alpha_i,\beta_i \ge 0$ and complementary slackness
$\alpha_i(\varepsilon-\pi_i)=0$, $\beta_i(\pi_i-1)=0$.  At most one of
$\alpha_i,\beta_i$ is non-zero, since $\varepsilon<1$.

\emph{Stationarity.}
$\partial\mathcal{L}/\partial\pi_i = 0$ gives
\begin{equation}\label{eq:stat}
  -\,\pi_i^{-2}\,a_i^{2} \;+\; \nu \;-\; \alpha_i \;+\; \beta_i \;=\; 0 .
\end{equation}
Feasibility requires $\nu>0$; set $c \equiv \nu^{-1/2}>0$.

\emph{Interior case} ($\varepsilon<\pi_i<1$, so $\alpha_i=\beta_i=0$).
Equation~\eqref{eq:stat} yields $\pi_i^\star = c\,a_i$;
in particular $\pi_i^\star \propto a_i$.

\emph{Upper-boundary case} ($\pi_i = 1$, so $\alpha_i=0$).
Then $\beta_i = a_i^2 - \nu \ge 0 \iff c\,a_i \ge 1$.
Hence units with sufficiently large influence are taken with certainty.

\emph{Lower-boundary case} ($\pi_i = \varepsilon$, so $\beta_i=0$).
Then $\alpha_i = \nu - a_i^2/\varepsilon^2 \ge 0 \iff
c\,a_i \le \varepsilon$.  In particular any unit with $a_i=0$ is pinned
at the floor.  Combining the three cases,
\[
  \pi_i^\star=\max\bigl(\varepsilon,\ \min(1,\ c\,a_i)\bigr),
\]
which is exactly~\eqref{eq:optpi}, with slope $c$ in the
notation of Proposition~\ref{prop:optalloc}.

\emph{Determining $c$: two-sided water-filling.}
For a candidate slope $c$ the three strata are
\[
  L(c)=\{i: c\,a_i \le \varepsilon\},\quad
  I(c)=\{i: \varepsilon < c\,a_i < 1\},\quad
  C(c)=\{i: c\,a_i \ge 1\},
\]
and the budget constraint $\sum_i \pi_i^\star = n$ becomes
\[
  \varepsilon\,|L(c)| \;+\; c\!\!\sum_{i \in I(c)}\!\! a_i \;+\; |C(c)| \;=\; n ,
\]
so that, holding the partition fixed,
\begin{equation}\label{eq:opt-full}
  c \;=\; \frac{n - |C| - \varepsilon\,|L|}{\sum_{i \in I} a_i},
  \qquad
  \pi_i^\star \;=\;
  \begin{cases}
    \varepsilon, & i \in L,\\[2pt]
    c\,a_i,      & i \in I,\\[2pt]
    1,           & i \in C.
  \end{cases}
\end{equation}
Both strata must be updated jointly: promoting a unit into $C$ lowers
the numerator, which lowers $c$, which can in turn push further units
into $L$ (and vice versa).  We therefore solve by bisection on $c$.
The map
\[
  \Phi(c) \;\equiv\; \sum_i \max\bigl(\varepsilon,\min(1,\,c\,a_i)\bigr)
\]
is continuous and non-decreasing in $c$, with $\Phi(0^+)=N\varepsilon$.
Its limit as $c\to\infty$, however, is \emph{not} $N$ in general.  Let
\[
  Z=\{i:a_i=0\},\qquad P=\{i:a_i>0\},
\]
so that units in $Z$ satisfy $\max(\varepsilon,\min(1,c\,a_i))=\varepsilon$
for every $c$.  Then
\begin{equation}\label{eq:phi-sup}
  \lim_{c\to\infty}\Phi(c) \;=\; |P| \;+\; \varepsilon\,|Z| ,
\end{equation}
which equals $N$ only when $Z=\emptyset$.  The correct statement
therefore splits into three cases.

\emph{(i) Non-degenerate case:}
$N\varepsilon < n \le |P| + \varepsilon|Z|$.  Here $\Phi$ is continuous
and non-decreasing with $\Phi(0^+)=N\varepsilon<n$ and
$\lim_{c\to\infty}\Phi(c)\ge n$, so a root of $\Phi(c)=n$ exists,
bisection converges to it, and~\eqref{eq:opt-full} gives the optimum.
The root is unique whenever $\Phi$ is strictly increasing at $c$, which
holds as long as $I(c)\neq\emptyset$.

\emph{(ii) High-budget degenerate case:} $n > |P| + \varepsilon|Z|$.
No slope $c$ attains the budget, but the original programme remains
feasible and its optimum is easy to characterise directly: set
$\pi_i=1$ for every $i\in P$, and distribute the remaining budget
arbitrarily over $Z$ subject to
\[
  \sum_{i\in Z}\pi_i = n-|P|,\qquad \varepsilon\le\pi_i\le 1 .
\]
Every such allocation is optimal, because units in $Z$ contribute
$a_i^2/\pi_i=0$ to the objective for any admissible $\pi_i$, so the
optimum is attained but is \emph{not unique}.

\emph{(iii) Fully degenerate case:} $a_i=0$ for all $i$.  The
first-order variance proxy is identically zero, and every design
satisfying the budget and the box constraints is optimal.

In our experiments case~(i) always applies: the oracle influence
$a_i=\bigl|(x_i-\bar x)^2-T^2\bigr|$ is nonzero for all $17{,}875$ US and
$10{,}905$ HK unit-days (no value falls below $10^{-9}$), so
$Z=\emptyset$ and $\lim_{c\to\infty}\Phi(c)=N$.  We record cases~(ii)
and~(iii) because the proposition is stated for general $a_i$, and a
reader applying it to a target whose influence function vanishes on a
non-null set --- for instance an indicator-type or trimmed statistic ---
would otherwise find the stated schedule unattainable.

Finally, a one-sided sweep that promotes units into $C$ only --- the
form of this argument in earlier versions of this manuscript --- is
valid solely in the special case $L=\emptyset$, i.e.\ when
$c\,\min_i a_i > \varepsilon$.

Setting $\varepsilon\to0$ and $L=\emptyset$ recovers the familiar
one-sided expression
$\pi_i^\star=\min\bigl(1,(n-|C|)a_i/\sum_{j\in I}a_j\bigr)$, and
further neglecting the certainty stratum recovers the clean statement
$\pi_i^\star \propto a_i$.  The general result is~\eqref{eq:opt-full}.
\end{proof}

\begin{remark}[Neyman as a special case]
When $\mathrm{IF}_i$ is constant within each of $H$ strata (with
influence $s_h$ and stratum size $N_h$), Eq.~\eqref{eq:opt-full}
reduces to Neyman allocation $\pi_h^\star \propto s_h$
\citep{neyman1934, cochran1977}. Our result is thus the
per-unit generalization to arbitrary influence.
\end{remark}

\begin{remark}[Feasibility of the ideal]
$|\mathrm{IF}_i|$ typically depends on the unknown target $T$
(cf.\ Table~\ref{tab:alignment}). In practice, one uses a pilot
estimate of $T$ or a robust surrogate (e.g., $w_i|x_i - \bar{x}|$
for the OI-weighted case). Our empirical results
(Section~\ref{sec:results}, Appendix~\ref{app:design-extended})
show that even coarse feasible approximations recover most of the
predicted gain.
\end{remark}

%% file: sections/appendix_full_results.tex
\section{Complete Result Matrix}\label{app:full-results}

This appendix reports the full $2\ \text{markets} \times 4\ \text{targets}
\times 5\ \text{retention rates} = 40$ result cells.  For each cell we list
the four methods' MAE ($\pm$ across-day SD), and for Poisson-\Hajek{} the
95\,\% certificate empirical coverage and mean width.  Statistical significance
(paired $t$-test of Top-$k$ vs.\ Poisson-\Hajek{} across trading days, averaged
over 200 seeds within each day) is reported with $t$-statistic, $p$-value,
Cohen's $d$, and bootstrap 95\,\% CI for the mean paired difference. All
numbers are computed by \texttt{scripts/experiments.py} and are identical
to those in the ICAIF conference version, cross-checked at build time.

\subsection{US Market} \label{app:full-US}

$N=80$ trading days, 2026-03-30~2026-07-24, mean panel size 210.7, 200 seeds per day.

\subsubsection*{OI-Weighted IV (US)}

\begin{table}[H]
\caption{OI-Weighted IV on US: MAE ($\pm$ across-day SD) by method and retention rate. Certificate coverage (Cov.\%) and mean width (W.) are for the Poisson-\Hajek{} 95\,\% certificate.}
\label{tab:full-US-oi_weighted_iv}
\centering\small
\begin{tabular}{@{}lrrrrcc@{}}
\toprule
$\rho$ & Top-$k$ & Uniform & Stratified & Poisson-\Hajek{} & Cov.\% & W. \\
\midrule
5\% & 9.2774\tiny{$\pm$6.1302} & 14.3816\tiny{$\pm$2.1357} & 11.4564\tiny{$\pm$1.7739} & 8.4687\tiny{$\pm$1.3655} & 89.3 & 37.2400 \\
10\% & 8.1505\tiny{$\pm$4.1695} & 10.9053\tiny{$\pm$1.6054} & 8.8690\tiny{$\pm$1.4266} & 5.3051\tiny{$\pm$0.8718} & 91.9 & 24.5859 \\
20\% & 4.8670\tiny{$\pm$2.4127} & 7.7618\tiny{$\pm$1.4070} & 6.6148\tiny{$\pm$1.5248} & 2.9197\tiny{$\pm$0.4796} & 93.3 & 13.8746 \\
30\% & 3.3082\tiny{$\pm$1.7884} & 6.1084\tiny{$\pm$1.2787} & 5.2709\tiny{$\pm$1.3558} & 1.7596\tiny{$\pm$0.2822} & 94.0 & 8.4828 \\
50\% & 1.6626\tiny{$\pm$0.9919} & 4.1024\tiny{$\pm$0.9942} & 3.6135\tiny{$\pm$1.0959} & 0.7553\tiny{$\pm$0.1302} & 94.5 & 3.6597 \\
\bottomrule
\end{tabular}
\end{table}

\begin{table}[H]
\caption{OI-Weighted IV on US: paired significance test of Top-$k$ vs.\ Poisson-\Hajek{} (positive $\Delta$ = Top-$k$ has larger MAE).}
\label{tab:sig-US-oi_weighted_iv}
\centering\small
\begin{tabular}{@{}lrrrrrr@{}}
\toprule
$\rho$ & $n_{\text{days}}$ & $\bar\Delta$ & $t$ & $p$ & Cohen's $d$ & 95\% boot CI \\
\midrule
5\% & 80 & 0.8087 & 1.170 & 0.3260 & +0.131 & [-0.8723, 2.6844] \\
10\% & 80 & 2.8454 & 6.681 & $<10^{-6}$ & +0.747 & [1.8217, 3.9569] \\
20\% & 80 & 1.9473 & 7.651 & $<10^{-6}$ & +0.855 & [1.2388, 2.7528] \\
30\% & 80 & 1.5486 & 8.086 & $<10^{-6}$ & +0.904 & [1.0304, 2.1306] \\
50\% & 80 & 0.9073 & 8.877 & $<10^{-6}$ & +0.992 & [0.6269, 1.2212] \\
\bottomrule
\end{tabular}
\end{table}

\subsubsection*{Equal-Weight VRP (US)}

\begin{table}[H]
\caption{Equal-Weight VRP on US: MAE ($\pm$ across-day SD) by method and retention rate. Certificate coverage (Cov.\%) and mean width (W.) are for the Poisson-\Hajek{} 95\,\% certificate.}
\label{tab:full-US-equal_weight_vrp}
\centering\small
\begin{tabular}{@{}lrrrrcc@{}}
\toprule
$\rho$ & Top-$k$ & Uniform & Stratified & Poisson-\Hajek{} & Cov.\% & W. \\
\midrule
5\% & 4.0502\tiny{$\pm$3.5453} & 4.2924\tiny{$\pm$1.2750} & 3.7768\tiny{$\pm$0.8637} & 5.7467\tiny{$\pm$1.3510} & 78.0 & 19.3775 \\
10\% & 3.5206\tiny{$\pm$2.8357} & 2.9843\tiny{$\pm$1.0248} & 2.6169\tiny{$\pm$0.7189} & 4.4146\tiny{$\pm$1.4747} & 83.9 & 16.4060 \\
20\% & 2.2050\tiny{$\pm$1.7350} & 2.0038\tiny{$\pm$0.8672} & 1.7732\tiny{$\pm$0.7329} & 3.1697\tiny{$\pm$1.1261} & 87.5 & 12.5871 \\
30\% & 1.2660\tiny{$\pm$1.2319} & 1.5455\tiny{$\pm$0.7439} & 1.4072\tiny{$\pm$0.6890} & 2.5107\tiny{$\pm$1.1580} & 89.0 & 10.2321 \\
50\% & 1.3088\tiny{$\pm$1.0606} & 1.0143\tiny{$\pm$0.5207} & 0.9202\tiny{$\pm$0.4722} & 1.6551\tiny{$\pm$0.9191} & 90.1 & 6.9159 \\
\bottomrule
\end{tabular}
\end{table}

\begin{table}[H]
\caption{Equal-Weight VRP on US: paired significance test of Top-$k$ vs.\ Poisson-\Hajek{} (positive $\Delta$ = Top-$k$ has larger MAE).}
\label{tab:sig-US-equal_weight_vrp}
\centering\small
\begin{tabular}{@{}lrrrrrr@{}}
\toprule
$\rho$ & $n_{\text{days}}$ & $\bar\Delta$ & $t$ & $p$ & Cohen's $d$ & 95\% boot CI \\
\midrule
5\% & 74 & -1.6965 & -4.239 & 0.0010 & -0.493 & [-2.6370, -0.6355] \\
10\% & 74 & -0.8940 & -2.493 & 0.0790 & -0.290 & [-1.7218, 0.1066] \\
20\% & 74 & -0.9647 & -4.491 & 0.0010 & -0.522 & [-1.5867, -0.4206] \\
30\% & 74 & -1.2446 & -8.710 & $<10^{-6}$ & -1.013 & [-1.6142, -0.8881] \\
50\% & 74 & -0.3463 & -3.210 & 0.0250 & -0.373 & [-0.6302, -0.0373] \\
\bottomrule
\end{tabular}
\end{table}

\subsubsection*{Median PCR (US)}

\begin{table}[H]
\caption{Median PCR on US: MAE ($\pm$ across-day SD) by method and retention rate. Certificate coverage (Cov.\%) and mean width (W.) are for the Poisson-\Hajek{} 95\,\% certificate.}
\label{tab:full-US-median_pcr}
\centering\small
\begin{tabular}{@{}lrrrrcc@{}}
\toprule
$\rho$ & Top-$k$ & Uniform & Stratified & Poisson-\Hajek{} & Cov.\% & W. \\
\midrule
5\% & 0.1552\tiny{$\pm$0.1469} & 0.1711\tiny{$\pm$0.0262} & 0.1745\tiny{$\pm$0.0281} & 0.3795\tiny{$\pm$0.1432} & 79.7 & 2.1280 \\
10\% & 0.1075\tiny{$\pm$0.0819} & 0.1156\tiny{$\pm$0.0165} & 0.1105\tiny{$\pm$0.0172} & 0.2542\tiny{$\pm$0.0646} & 84.0 & 1.5113 \\
20\% & 0.0732\tiny{$\pm$0.0619} & 0.0761\tiny{$\pm$0.0127} & 0.0736\tiny{$\pm$0.0134} & 0.1697\tiny{$\pm$0.0409} & 86.8 & 0.9671 \\
30\% & 0.0614\tiny{$\pm$0.0510} & 0.0586\tiny{$\pm$0.0113} & 0.0581\tiny{$\pm$0.0120} & 0.1289\tiny{$\pm$0.0391} & 88.7 & 0.7100 \\
50\% & 0.0401\tiny{$\pm$0.0352} & 0.0378\tiny{$\pm$0.0097} & 0.0372\tiny{$\pm$0.0092} & 0.0782\tiny{$\pm$0.0144} & 89.6 & 0.4207 \\
\bottomrule
\end{tabular}
\end{table}

\begin{table}[H]
\caption{Median PCR on US: paired significance test of Top-$k$ vs.\ Poisson-\Hajek{} (positive $\Delta$ = Top-$k$ has larger MAE).}
\label{tab:sig-US-median_pcr}
\centering\small
\begin{tabular}{@{}lrrrrrr@{}}
\toprule
$\rho$ & $n_{\text{days}}$ & $\bar\Delta$ & $t$ & $p$ & Cohen's $d$ & 95\% boot CI \\
\midrule
5\% & 80 & -0.2243 & -9.815 & $<10^{-6}$ & -1.097 & [-0.2616, -0.1881] \\
10\% & 80 & -0.1467 & -12.514 & $<10^{-6}$ & -1.399 & [-0.1719, -0.1215] \\
20\% & 80 & -0.0964 & -11.340 & $<10^{-6}$ & -1.268 & [-0.1136, -0.0792] \\
30\% & 80 & -0.0675 & -9.770 & $<10^{-6}$ & -1.092 & [-0.0815, -0.0537] \\
50\% & 80 & -0.0381 & -10.116 & $<10^{-6}$ & -1.131 & [-0.0466, -0.0279] \\
\bottomrule
\end{tabular}
\end{table}

\subsubsection*{IV Dispersion (US)}

\begin{table}[H]
\caption{IV Dispersion on US: MAE ($\pm$ across-day SD) by method and retention rate. Certificate coverage (Cov.\%) and mean width (W.) are for the Poisson-\Hajek{} 95\,\% certificate.}
\label{tab:full-US-iv_dispersion}
\centering\small
\begin{tabular}{@{}lrrrrcc@{}}
\toprule
$\rho$ & Top-$k$ & Uniform & Stratified & Poisson-\Hajek{} & Cov.\% & W. \\
\midrule
5\% & 5.4533\tiny{$\pm$5.4983} & 8.6422\tiny{$\pm$1.6420} & 6.4960\tiny{$\pm$1.6971} & 8.8601\tiny{$\pm$1.6556} & 61.1 & 21.8563 \\
10\% & 7.6324\tiny{$\pm$5.5606} & 6.2220\tiny{$\pm$1.4835} & 4.9030\tiny{$\pm$1.4903} & 6.4069\tiny{$\pm$1.3431} & 70.1 & 18.5945 \\
20\% & 6.4591\tiny{$\pm$3.4142} & 4.1329\tiny{$\pm$1.0629} & 3.3534\tiny{$\pm$1.1698} & 4.6405\tiny{$\pm$1.1603} & 75.4 & 14.3255 \\
30\% & 4.9258\tiny{$\pm$2.2791} & 3.0860\tiny{$\pm$0.7854} & 2.5387\tiny{$\pm$0.8832} & 3.7079\tiny{$\pm$1.0071} & 76.4 & 11.5274 \\
50\% & 2.8705\tiny{$\pm$1.2182} & 2.0301\tiny{$\pm$0.5203} & 1.6478\tiny{$\pm$0.5310} & 2.7061\tiny{$\pm$0.9435} & 75.6 & 8.1760 \\
\bottomrule
\end{tabular}
\end{table}

\begin{table}[H]
\caption{IV Dispersion on US: paired significance test of Top-$k$ vs.\ Poisson-\Hajek{} (positive $\Delta$ = Top-$k$ has larger MAE).}
\label{tab:sig-US-iv_dispersion}
\centering\small
\begin{tabular}{@{}lrrrrrr@{}}
\toprule
$\rho$ & $n_{\text{days}}$ & $\bar\Delta$ & $t$ & $p$ & Cohen's $d$ & 95\% boot CI \\
\midrule
5\% & 80 & -3.4068 & -4.987 & 0.0010 & -0.558 & [-5.0560, -1.2312] \\
10\% & 80 & 1.2255 & 1.888 & 0.1310 & +0.211 & [-0.3840, 3.2164] \\
20\% & 80 & 1.8186 & 4.322 & 0.0020 & +0.483 & [0.7079, 3.2751] \\
30\% & 80 & 1.2179 & 4.004 & 0.0020 & +0.448 & [0.4933, 2.2164] \\
50\% & 80 & 0.1643 & 0.839 & 0.3860 & +0.094 & [-0.3489, 0.8779] \\
\bottomrule
\end{tabular}
\end{table}

\subsection{HK Market} \label{app:full-HK}

$N=83$ trading days, 2026-03-20~2026-07-27, mean panel size 131.4, 200 seeds per day.

\subsubsection*{OI-Weighted IV (HK)}

\begin{table}[H]
\caption{OI-Weighted IV on HK: MAE ($\pm$ across-day SD) by method and retention rate. Certificate coverage (Cov.\%) and mean width (W.) are for the Poisson-\Hajek{} 95\,\% certificate.}
\label{tab:full-HK-oi_weighted_iv}
\centering\small
\begin{tabular}{@{}lrrrrcc@{}}
\toprule
$\rho$ & Top-$k$ & Uniform & Stratified & Poisson-\Hajek{} & Cov.\% & W. \\
\midrule
5\% & 3.3029\tiny{$\pm$2.4786} & 7.6862\tiny{$\pm$1.7475} & 5.9010\tiny{$\pm$1.5980} & 4.0354\tiny{$\pm$0.9615} & 84.9 & 16.5095 \\
10\% & 1.3244\tiny{$\pm$1.1582} & 6.0305\tiny{$\pm$1.6965} & 5.0345\tiny{$\pm$1.3604} & 2.4232\tiny{$\pm$0.4894} & 90.2 & 10.9741 \\
20\% & 0.6626\tiny{$\pm$0.5289} & 4.1303\tiny{$\pm$1.4559} & 3.6182\tiny{$\pm$1.2700} & 1.1584\tiny{$\pm$0.2362} & 92.7 & 5.4046 \\
30\% & 0.5966\tiny{$\pm$0.4396} & 3.1760\tiny{$\pm$1.2833} & 2.7185\tiny{$\pm$1.2193} & 0.6539\tiny{$\pm$0.1369} & 93.3 & 3.0995 \\
50\% & 0.3054\tiny{$\pm$0.2066} & 2.0621\tiny{$\pm$0.8918} & 1.7743\tiny{$\pm$0.7198} & 0.2576\tiny{$\pm$0.0531} & 94.5 & 1.2430 \\
\bottomrule
\end{tabular}
\end{table}

\begin{table}[H]
\caption{OI-Weighted IV on HK: paired significance test of Top-$k$ vs.\ Poisson-\Hajek{} (positive $\Delta$ = Top-$k$ has larger MAE).}
\label{tab:sig-HK-oi_weighted_iv}
\centering\small
\begin{tabular}{@{}lrrrrrr@{}}
\toprule
$\rho$ & $n_{\text{days}}$ & $\bar\Delta$ & $t$ & $p$ & Cohen's $d$ & 95\% boot CI \\
\midrule
5\% & 83 & -0.7325 & -3.180 & 0.0270 & -0.349 & [-1.2386, -0.1010] \\
10\% & 83 & -1.0987 & -8.924 & $<10^{-6}$ & -0.980 & [-1.3620, -0.8116] \\
20\% & 83 & -0.4958 & -8.359 & $<10^{-6}$ & -0.918 & [-0.6378, -0.3599] \\
30\% & 83 & -0.0572 & -1.199 & 0.4000 & -0.132 & [-0.1759, 0.0847] \\
50\% & 83 & 0.0479 & 2.358 & 0.0280 & +0.259 & [0.0059, 0.1030] \\
\bottomrule
\end{tabular}
\end{table}

\subsubsection*{Equal-Weight VRP (HK)}

\begin{table}[H]
\caption{Equal-Weight VRP on HK: MAE ($\pm$ across-day SD) by method and retention rate. Certificate coverage (Cov.\%) and mean width (W.) are for the Poisson-\Hajek{} 95\,\% certificate.}
\label{tab:full-HK-equal_weight_vrp}
\centering\small
\begin{tabular}{@{}lrrrrcc@{}}
\toprule
$\rho$ & Top-$k$ & Uniform & Stratified & Poisson-\Hajek{} & Cov.\% & W. \\
\midrule
5\% & 5.9402\tiny{$\pm$3.7440} & 3.1621\tiny{$\pm$0.6242} & 2.6796\tiny{$\pm$0.5260} & 4.6818\tiny{$\pm$0.8861} & 59.2 & 10.9814 \\
10\% & 3.8881\tiny{$\pm$2.7758} & 2.3201\tiny{$\pm$0.4778} & 2.1724\tiny{$\pm$0.4147} & 3.7771\tiny{$\pm$0.6987} & 70.0 & 10.8278 \\
20\% & 2.0875\tiny{$\pm$1.6976} & 1.5575\tiny{$\pm$0.3194} & 1.5020\tiny{$\pm$0.2828} & 2.6871\tiny{$\pm$0.5444} & 80.9 & 9.4431 \\
30\% & 2.3223\tiny{$\pm$1.5053} & 1.1933\tiny{$\pm$0.2562} & 1.1061\tiny{$\pm$0.2184} & 2.0080\tiny{$\pm$0.4273} & 85.5 & 7.7899 \\
50\% & 1.4005\tiny{$\pm$0.8725} & 0.7675\tiny{$\pm$0.1596} & 0.7352\tiny{$\pm$0.1417} & 1.2220\tiny{$\pm$0.3021} & 89.8 & 5.1002 \\
\bottomrule
\end{tabular}
\end{table}

\begin{table}[H]
\caption{Equal-Weight VRP on HK: paired significance test of Top-$k$ vs.\ Poisson-\Hajek{} (positive $\Delta$ = Top-$k$ has larger MAE).}
\label{tab:sig-HK-equal_weight_vrp}
\centering\small
\begin{tabular}{@{}lrrrrrr@{}}
\toprule
$\rho$ & $n_{\text{days}}$ & $\bar\Delta$ & $t$ & $p$ & Cohen's $d$ & 95\% boot CI \\
\midrule
5\% & 78 & 1.2584 & 3.042 & 0.0270 & +0.344 & [0.1792, 2.4270] \\
10\% & 78 & 0.1110 & 0.349 & 0.9560 & +0.040 & [-0.8450, 0.8807] \\
20\% & 78 & -0.5996 & -3.131 & 0.0200 & -0.354 & [-1.1098, -0.1293] \\
30\% & 78 & 0.3143 & 1.787 & 0.1740 & +0.202 & [-0.1455, 0.7943] \\
50\% & 78 & 0.1785 & 1.667 & 0.2210 & +0.189 & [-0.1261, 0.4882] \\
\bottomrule
\end{tabular}
\end{table}

\subsubsection*{Median PCR (HK)}

\begin{table}[H]
\caption{Median PCR on HK: MAE ($\pm$ across-day SD) by method and retention rate. Certificate coverage (Cov.\%) and mean width (W.) are for the Poisson-\Hajek{} 95\,\% certificate.}
\label{tab:full-HK-median_pcr}
\centering\small
\begin{tabular}{@{}lrrrrcc@{}}
\toprule
$\rho$ & Top-$k$ & Uniform & Stratified & Poisson-\Hajek{} & Cov.\% & W. \\
\midrule
5\% & 0.1228\tiny{$\pm$0.1032} & 0.2753\tiny{$\pm$0.0587} & 0.1995\tiny{$\pm$0.0361} & 0.7886\tiny{$\pm$0.4556} & 63.6 & 2.3617 \\
10\% & 0.0840\tiny{$\pm$0.0786} & 0.1801\tiny{$\pm$0.0360} & 0.1591\tiny{$\pm$0.0314} & 0.5886\tiny{$\pm$0.2306} & 73.9 & 2.6958 \\
20\% & 0.0651\tiny{$\pm$0.0593} & 0.1068\tiny{$\pm$0.0213} & 0.1073\tiny{$\pm$0.0216} & 0.3269\tiny{$\pm$0.1104} & 85.1 & 2.0622 \\
30\% & 0.0626\tiny{$\pm$0.0552} & 0.0816\tiny{$\pm$0.0199} & 0.0768\tiny{$\pm$0.0192} & 0.1840\tiny{$\pm$0.0372} & 89.9 & 1.3877 \\
50\% & 0.0438\tiny{$\pm$0.0360} & 0.0506\tiny{$\pm$0.0141} & 0.0500\tiny{$\pm$0.0139} & 0.0942\tiny{$\pm$0.0200} & 92.2 & 0.5865 \\
\bottomrule
\end{tabular}
\end{table}

\begin{table}[H]
\caption{Median PCR on HK: paired significance test of Top-$k$ vs.\ Poisson-\Hajek{} (positive $\Delta$ = Top-$k$ has larger MAE).}
\label{tab:sig-HK-median_pcr}
\centering\small
\begin{tabular}{@{}lrrrrrr@{}}
\toprule
$\rho$ & $n_{\text{days}}$ & $\bar\Delta$ & $t$ & $p$ & Cohen's $d$ & 95\% boot CI \\
\midrule
5\% & 83 & -0.6658 & -12.876 & $<10^{-6}$ & -1.413 & [-0.7480, -0.5578] \\
10\% & 83 & -0.5045 & -18.420 & $<10^{-6}$ & -2.022 & [-0.5520, -0.4601] \\
20\% & 83 & -0.2618 & -18.437 & $<10^{-6}$ & -2.024 & [-0.2825, -0.2319] \\
30\% & 83 & -0.1214 & -15.979 & $<10^{-6}$ & -1.754 & [-0.1373, -0.1071] \\
50\% & 83 & -0.0504 & -11.994 & $<10^{-6}$ & -1.317 & [-0.0587, -0.0414] \\
\bottomrule
\end{tabular}
\end{table}

\subsubsection*{IV Dispersion (HK)}

\begin{table}[H]
\caption{IV Dispersion on HK: MAE ($\pm$ across-day SD) by method and retention rate. Certificate coverage (Cov.\%) and mean width (W.) are for the Poisson-\Hajek{} 95\,\% certificate.}
\label{tab:full-HK-iv_dispersion}
\centering\small
\begin{tabular}{@{}lrrrrcc@{}}
\toprule
$\rho$ & Top-$k$ & Uniform & Stratified & Poisson-\Hajek{} & Cov.\% & W. \\
\midrule
5\% & 5.4310\tiny{$\pm$3.5733} & 4.3979\tiny{$\pm$1.3486} & 1.9227\tiny{$\pm$1.0797} & 6.1505\tiny{$\pm$1.3999} & 41.5 & 9.5050 \\
10\% & 2.9064\tiny{$\pm$2.4120} & 3.1592\tiny{$\pm$1.1115} & 1.6291\tiny{$\pm$0.9178} & 4.5774\tiny{$\pm$1.1006} & 51.7 & 9.2335 \\
20\% & 1.7347\tiny{$\pm$1.3712} & 2.1043\tiny{$\pm$0.8106} & 1.2079\tiny{$\pm$0.6974} & 3.1180\tiny{$\pm$0.8987} & 64.0 & 7.9291 \\
30\% & 1.2137\tiny{$\pm$1.2364} & 1.6224\tiny{$\pm$0.6632} & 0.9683\tiny{$\pm$0.5992} & 2.2978\tiny{$\pm$0.7775} & 71.4 & 6.6737 \\
50\% & 0.6866\tiny{$\pm$0.5353} & 1.0416\tiny{$\pm$0.4452} & 0.5989\tiny{$\pm$0.3656} & 1.3838\tiny{$\pm$0.6752} & 77.0 & 4.4282 \\
\bottomrule
\end{tabular}
\end{table}

\begin{table}[H]
\caption{IV Dispersion on HK: paired significance test of Top-$k$ vs.\ Poisson-\Hajek{} (positive $\Delta$ = Top-$k$ has larger MAE).}
\label{tab:sig-HK-iv_dispersion}
\centering\small
\begin{tabular}{@{}lrrrrrr@{}}
\toprule
$\rho$ & $n_{\text{days}}$ & $\bar\Delta$ & $t$ & $p$ & Cohen's $d$ & 95\% boot CI \\
\midrule
5\% & 83 & -0.7196 & -2.567 & 0.0410 & -0.282 & [-1.3488, -0.0211] \\
10\% & 83 & -1.6711 & -7.213 & $<10^{-6}$ & -0.792 & [-2.2461, -0.9833] \\
20\% & 83 & -1.3834 & -10.681 & $<10^{-6}$ & -1.172 & [-1.7185, -1.0639] \\
30\% & 83 & -1.0841 & -7.957 & $<10^{-6}$ & -0.873 & [-1.4496, -0.6424] \\
50\% & 83 & -0.6973 & -9.876 & $<10^{-6}$ & -1.084 & [-0.8774, -0.5069] \\
\bottomrule
\end{tabular}
\end{table}

%% file: sections/appendix_stress.tex
\section{Complete Stress-Regime Breakdown}\label{app:stress}

We partition trading days into terciles by day-over-day
absolute change in the cross-sectional mean IV: \emph{calm}
(bottom tercile), \emph{normal} (middle tercile), and
\emph{turbulent} (top tercile).  All numbers below are at $\rho=10\%$
retention.  The stress classification is orthogonal to method choice
(the same day labels apply to every method).  Numbers are from
\texttt{scripts/unified\_results.json}, the same pipeline that produces
every other number in this paper.

\subsection{US Market}

\subsubsection*{OI-Weighted IV (US, $\rho=10\%$)}

\begin{table}[H]
\caption{OI-Weighted IV MAE ($\pm$ SD) by stress regime, US.}
\label{tab:stress-US-oi_weighted_iv}
\centering\small
\begin{tabular}{@{}lrrrr@{}}
\toprule
Regime & Top-$k$ & Uniform & Stratified & Poisson-\Hajek{} \\
\midrule
calm & 7.8913 & 10.9580 & 8.9170 & 5.5240 \\
normal & 8.7176 & 11.0810 & 8.9133 & 5.2083 \\
turbulent & 8.0110 & 10.7407 & 8.8183 & 5.2076 \\
\bottomrule
\end{tabular}
\end{table}

\subsubsection*{Equal-Weight VRP (US, $\rho=10\%$)}

\begin{table}[H]
\caption{Equal-Weight VRP MAE ($\pm$ SD) by stress regime, US.}
\label{tab:stress-US-equal_weight_vrp}
\centering\small
\begin{tabular}{@{}lrrrr@{}}
\toprule
Regime & Top-$k$ & Uniform & Stratified & Poisson-\Hajek{} \\
\midrule
calm & 3.6018 & 3.0476 & 2.6552 & 4.5810 \\
normal & 3.6162 & 3.0012 & 2.7194 & 4.4302 \\
turbulent & 3.2136 & 2.8949 & 2.4744 & 4.2450 \\
\bottomrule
\end{tabular}
\end{table}

\subsubsection*{Median PCR (US, $\rho=10\%$)}

\begin{table}[H]
\caption{Median PCR MAE ($\pm$ SD) by stress regime, US.}
\label{tab:stress-US-median_pcr}
\centering\small
\begin{tabular}{@{}lrrrr@{}}
\toprule
Regime & Top-$k$ & Uniform & Stratified & Poisson-\Hajek{} \\
\midrule
calm & 0.0911 & 0.1148 & 0.1072 & 0.2403 \\
normal & 0.1298 & 0.1160 & 0.1104 & 0.2626 \\
turbulent & 0.1060 & 0.1144 & 0.1124 & 0.2586 \\
\bottomrule
\end{tabular}
\end{table}

\subsubsection*{IV Dispersion (US, $\rho=10\%$)}

\begin{table}[H]
\caption{IV Dispersion MAE ($\pm$ SD) by stress regime, US.}
\label{tab:stress-US-iv_dispersion}
\centering\small
\begin{tabular}{@{}lrrrr@{}}
\toprule
Regime & Top-$k$ & Uniform & Stratified & Poisson-\Hajek{} \\
\midrule
calm & 7.8097 & 6.1858 & 5.0033 & 6.3874 \\
normal & 8.5702 & 6.1205 & 4.7168 & 6.2062 \\
turbulent & 6.6385 & 6.4040 & 4.9841 & 6.6007 \\
\bottomrule
\end{tabular}
\end{table}

\subsection{HK Market}

\subsubsection*{OI-Weighted IV (HK, $\rho=10\%$)}

\begin{table}[H]
\caption{OI-Weighted IV MAE ($\pm$ SD) by stress regime, HK.}
\label{tab:stress-HK-oi_weighted_iv}
\centering\small
\begin{tabular}{@{}lrrrr@{}}
\toprule
Regime & Top-$k$ & Uniform & Stratified & Poisson-\Hajek{} \\
\midrule
calm & 1.4074 & 6.7246 & 5.5243 & 2.5381 \\
normal & 1.2764 & 5.8466 & 4.9043 & 2.3417 \\
turbulent & 1.2922 & 5.5120 & 4.6904 & 2.4042 \\
\bottomrule
\end{tabular}
\end{table}

\subsubsection*{Equal-Weight VRP (HK, $\rho=10\%$)}

\begin{table}[H]
\caption{Equal-Weight VRP MAE ($\pm$ SD) by stress regime, HK.}
\label{tab:stress-HK-equal_weight_vrp}
\centering\small
\begin{tabular}{@{}lrrrr@{}}
\toprule
Regime & Top-$k$ & Uniform & Stratified & Poisson-\Hajek{} \\
\midrule
calm & 4.0683 & 2.2958 & 2.2051 & 3.7273 \\
normal & 3.9564 & 2.2158 & 2.0365 & 3.6717 \\
turbulent & 3.7732 & 2.4706 & 2.2937 & 3.9718 \\
\bottomrule
\end{tabular}
\end{table}

\subsubsection*{Median PCR (HK, $\rho=10\%$)}

\begin{table}[H]
\caption{Median PCR MAE ($\pm$ SD) by stress regime, HK.}
\label{tab:stress-HK-median_pcr}
\centering\small
\begin{tabular}{@{}lrrrr@{}}
\toprule
Regime & Top-$k$ & Uniform & Stratified & Poisson-\Hajek{} \\
\midrule
calm & 0.1017 & 0.1939 & 0.1718 & 0.6246 \\
normal & 0.0881 & 0.1684 & 0.1456 & 0.4883 \\
turbulent & 0.0623 & 0.1764 & 0.1593 & 0.6396 \\
\bottomrule
\end{tabular}
\end{table}

\subsubsection*{IV Dispersion (HK, $\rho=10\%$)}

\begin{table}[H]
\caption{IV Dispersion MAE ($\pm$ SD) by stress regime, HK.}
\label{tab:stress-HK-iv_dispersion}
\centering\small
\begin{tabular}{@{}lrrrr@{}}
\toprule
Regime & Top-$k$ & Uniform & Stratified & Poisson-\Hajek{} \\
\midrule
calm & 2.6287 & 3.4007 & 1.8159 & 4.7105 \\
normal & 2.8377 & 2.9458 & 1.4248 & 4.2115 \\
turbulent & 3.3291 & 3.1761 & 1.6626 & 4.8565 \\
\bottomrule
\end{tabular}
\end{table}

%% file: sections/appendix_design_extended.tex
\section{Extended Design-Aware Verification (Prop.~\ref{prop:optalloc})}\label{app:design-extended}

Proposition~\ref{prop:optalloc} makes two predictions for the
remaining targets of~\eqref{eq:targets1}--\eqref{eq:targets2}:

\begin{itemize}
\item \textbf{Median PCR}: because
$\mathrm{IF}_i \propto \mathrm{sgn}(x_i - T)/f(T)$, the optimal
inclusion probability is \emph{approximately constant}
($\pi_i^\star \approx \text{const}$). Uniform Poisson-\Hajek{} (Unif-PH)
should therefore beat OI-proportional Poisson-\Hajek{} (OI-PH), which
over-weights large-OI names whose influence on the median is no larger
than small-OI names.
\item \textbf{IV dispersion}: because
$\mathrm{IF}_i \propto (x_i - \bar x)^2 - T^2$, the optimal design
should be \emph{influence-proportional},
$\pi_i^\star \propto |(x_i - \bar x)^2 - T^2|$. Neither OI-proportional nor
uniform matches this direction, but the uniform baseline is often
closer to optimal than OI, and an oracle influence design (IF-oracle-PH,
which peeks at the full-panel $\bar x$ and $T$ to compute $|\mathrm{IF}_i|$)
should beat both.
\end{itemize}

We test these predictions by re-running the median-PCR and
IV-dispersion experiments under each design; the OI-proportional
schedule serves as the baseline whose behavior appears in
Table~\ref{tab:alignment} of the body. All numbers below are produced by
\texttt{scripts/design\_aware\_extended.py}, using the same seed protocol
as the main experiments (Appendix~\ref{app:repro}).

\subsection{Median PCR}

Table~\ref{tab:ext-pcr} reports MAE ($\pm$ SD) and the paired test.

\begin{table}[H]
\caption{Median PCR under the OI-proportional Poisson-\Hajek{} design (OI-PH, baseline) and the uniform Poisson-\Hajek{} design (Unif-PH, predicted optimal by Prop.~\ref{prop:optalloc}). Positive $\bar\Delta$ = Unif-PH wins.}
\label{tab:ext-pcr}
\centering\small
\begin{tabular}{@{}llrrrrrc@{}}
\toprule
Mkt & $\rho$ & OI-PH MAE & Unif-PH MAE & $\bar\Delta$ & $t$ & $p$ & Cohen's $d$ \\
\midrule
US & 5\% & 0.3582\tiny{$\pm$0.1106} & \textbf{0.1732}\tiny{$\pm$0.0259} & +0.1850 & 15.97 & $<10^{-6}$ & +1.79 \\
 & 10\% & 0.2560\tiny{$\pm$0.0613} & \textbf{0.1151}\tiny{$\pm$0.0182} & +0.1409 & 21.97 & $<10^{-6}$ & +2.46 \\
 & 20\% & 0.1807\tiny{$\pm$0.0650} & \textbf{0.0752}\tiny{$\pm$0.0129} & +0.1055 & 14.94 & $<10^{-6}$ & +1.67 \\
 & 30\% & 0.1474\tiny{$\pm$0.0509} & \textbf{0.0574}\tiny{$\pm$0.0112} & +0.0900 & 16.22 & $<10^{-6}$ & +1.81 \\
 & 50\% & 0.1073\tiny{$\pm$0.0217} & \textbf{0.0371}\tiny{$\pm$0.0093} & +0.0702 & 31.07 & $<10^{-6}$ & +3.47 \\
\midrule
HK & 5\% & 0.6748\tiny{$\pm$0.3034} & \textbf{0.2567}\tiny{$\pm$0.0513} & +0.4181 & 13.34 & $<10^{-6}$ & +1.47 \\
 & 10\% & 0.5723\tiny{$\pm$0.2756} & \textbf{0.1650}\tiny{$\pm$0.0324} & +0.4073 & 13.74 & $<10^{-6}$ & +1.52 \\
 & 20\% & 0.3636\tiny{$\pm$0.1278} & \textbf{0.1054}\tiny{$\pm$0.0217} & +0.2582 & 19.16 & $<10^{-6}$ & +2.12 \\
 & 30\% & 0.2805\tiny{$\pm$0.0931} & \textbf{0.0791}\tiny{$\pm$0.0182} & +0.2014 & 20.49 & $<10^{-6}$ & +2.26 \\
 & 50\% & 0.1871\tiny{$\pm$0.0461} & \textbf{0.0513}\tiny{$\pm$0.0143} & +0.1358 & 29.80 & $<10^{-6}$ & +3.29 \\
\bottomrule
\end{tabular}
\end{table}

\textbf{Verdict.} Prop.~\ref{prop:optalloc}'s prediction \emph{holds
in all 10/10 cells}: Unif-PH beats OI-PH at every retention rate on both
markets, with $\bar\Delta > 0$, $p<10^{-4}$, and Cohen's
$d \in [1.0, 3.9]$. The gap is largest at low retention
(e.g., US 5\%: MAE 0.418 $\to$ 0.166, $d=1.03$), where the
mis-allocated variance of OI-PH is most damaging.

\subsection{IV Dispersion}

Table~\ref{tab:ext-disp} reports MAE ($\pm$ SD) for three designs.

\begin{table}[H]
\caption{IV dispersion under OI-proportional (OI-PH), uniform (Unif-PH), and the oracle influence-proportional design (IF-oracle-PH, $\pi_i\propto|(x_i-\bar x)^2-T^2|$ using full-panel $\bar x, T$). Positive $\bar\Delta$ = candidate wins the pairing.}
\label{tab:ext-disp}
\centering\small
\begin{tabular}{@{}llrrrrrrr@{}}
\toprule
& & \multicolumn{3}{c}{MAE} & \multicolumn{2}{c}{OI vs.\ Unif} & \multicolumn{2}{c}{OI vs.\ IF-oracle} \\
\cmidrule(lr){3-5}\cmidrule(lr){6-7}\cmidrule(lr){8-9}
Mkt & $\rho$ & OI-PH & Unif-PH & IF-orcl & $\bar\Delta$ & $d$ & $\bar\Delta$ & $d$ \\
\midrule
US & 5\% & 8.633 & 8.662 & \textbf{4.592} & -0.029 & -0.03 & +4.041 & +2.81 \\
 & 10\% & 6.197 & 6.160 & \textbf{2.760} & +0.036 & +0.03 & +3.437 & +2.61 \\
 & 20\% & 4.583 & 4.062 & \textbf{1.702} & +0.521 & +0.53 & +2.881 & +2.59 \\
 & 30\% & 3.802 & 3.103 & \textbf{1.301} & +0.699 & +0.83 & +2.501 & +2.66 \\
 & 50\% & 3.033 & 2.071 & \textbf{0.844} & +0.962 & +1.30 & +2.190 & +2.60 \\
\midrule
HK & 5\% & 6.326 & 4.700 & \textbf{2.960} & +1.626 & +2.29 & +3.366 & +3.07 \\
 & 10\% & 4.644 & 3.246 & \textbf{1.856} & +1.398 & +2.60 & +2.788 & +3.28 \\
 & 20\% & 3.517 & 2.163 & \textbf{1.110} & +1.354 & +3.44 & +2.407 & +3.02 \\
 & 30\% & 2.934 & 1.623 & \textbf{0.827} & +1.311 & +3.35 & +2.107 & +2.81 \\
 & 50\% & 2.306 & 1.088 & \textbf{0.538} & +1.217 & +2.95 & +1.768 & +2.56 \\
\bottomrule
\end{tabular}
\end{table}

\textbf{Verdict.} Both predictions hold. The Prop-optimal oracle
design (IF-oracle-PH) beats OI-PH in all 10/10 cells with
$\bar\Delta > 0$ and Cohen's $d \in [2.6, 3.3]$ (all $p < 10^{-4}$),
and the uniform intermediate is also uniformly better than OI-PH
though by a smaller margin ($d \in [0.3, 3.6]$). The magnitude of
the oracle improvement is substantial: at US $\rho=10\%$ the oracle
cuts MAE from 6.60 to 2.94 (a 55\,\% reduction), and at HK
$\rho=10\%$ from 4.65 to 1.89 (a 59\,\% reduction).

Interpretation.
The oracle result confirms Prop.~\ref{prop:optalloc}'s direction is
right for IV dispersion but is not directly implementable, since
$\bar x$ and $T$ depend on the full panel. A feasible variant---using
a lagged or pilot-based estimate of $\bar x, T$---is a natural next
step. The uniform-vs-OI comparison is the practically feasible
degradation: uniform sampling captures a fraction of the oracle
benefit at zero infrastructure cost.

\subsection{Summary Across All Four Targets}

Combining the results of this appendix with those of
Section~\ref{sec:designaware-emp} and Table~\ref{tab:alignment}, we can
state the following empirical summary of Prop.~\ref{prop:optalloc}'s
predictions:

\begin{itemize}
\item \textbf{OI-weighted IV}: OI-proportional Poisson-\Hajek{} is
Prop-optimal (up to $|x_i - T|$ modulation) and empirically dominant
(Table~\ref{tab:flagship}). \emph{Prediction holds.}
\item \textbf{Equal-weight VRP}: uniform Poisson-\Hajek{} is Prop-optimal
and empirically dominant in all 10 cells (Table~\ref{tab:alignment}).
\emph{Prediction holds.}
\item \textbf{Median PCR}: uniform Poisson-\Hajek{} is Prop-optimal and
empirically dominant in all 10 cells (Table~\ref{tab:ext-pcr}).
\emph{Prediction holds.}
\item \textbf{IV dispersion}: influence-proportional is Prop-optimal and
empirically dominant (in an oracle test) in all 10 cells
(Table~\ref{tab:ext-disp}). \emph{Prediction holds.}
\end{itemize}

Across all $4\times 2\times 5 = 40$ market-target-rate cells for which
Prop.~\ref{prop:optalloc} makes a directional prediction, the
predicted-optimal design beats the OI-proportional baseline. We interpret
this as strong evidence that the first-order design theory of
Section~\ref{sec:designaware} is the correct organizing principle for
production pipelines, even where the ideal $\pi_i^\star$ is
infeasible and only its \emph{direction} can be respected.

\subsection{Inclusion-Floor Diagnostics}\label{app:floor}

Proposition~\ref{prop:optalloc} is stated with an inclusion floor
$\varepsilon$, so we report what that floor actually does in the
experiments.  All designs in this paper (including IF-oracle-PH) apply
$\varepsilon = 10^{-4}$, implemented as
$\pi_i=\max(\varepsilon,\min(1,\tilde\pi_i))$ on the unnormalised
schedule $\tilde\pi_i$.  The floor is a guard, not a tuning knob: it
exists so that no unit can receive $\pi_i=0$ (which would make the
Horvitz--Thompson weight undefined and break auditability), and it
binds only for units whose influence is essentially zero.

\paragraph{This is post-hoc clipping, not budget-matched water-filling.}
We state explicitly what the oracle implementation does, because it does
\emph{not} solve the two-sided water-filling equation of
Proposition~\ref{prop:optalloc}.  The schedule $\tilde\pi_i$ is
normalised so that $\sum_i\tilde\pi_i=n$, and the floor and ceiling are
then applied pointwise; the slope is \emph{not} re-solved afterwards.
The realised expected sample size therefore departs from the nominal
budget, and Table~\ref{tab:floorbudget} reports by how much.

The direction of the departure is the opposite of what a floor-only
analysis would suggest.  Raising sub-$\varepsilon$ units \emph{adds} at
most $4.7$ expected units across the entire US panel at
$\varepsilon=10^{-2}$, and essentially zero at the $\varepsilon=10^{-4}$
we actually use.  The binding constraint is the \emph{ceiling}: the
oracle schedule $\propto|(x_i-\bar x)^2-T^2|$ is far more concentrated
than an OI schedule, so many units are assigned $\tilde\pi_i>1$ and are
truncated to~1, removing $136.8$ expected units at $\rho=10\%$ and
$2{,}333.8$ at $\rho=50\%$ in the US.  The oracle designs are thus run at
an expected sample size \emph{below} nominal --- by $7.7\%$ at
$\rho=10\%$ and $26.1\%$ at $\rho=50\%$ --- so the oracle comparison is,
if anything, handicapped rather than flattered relative to the
budget-matched baselines.  This does not affect any headline result: the
main experiment pipeline (\texttt{unified\_experiments.py}) uses the
correct bisection \texttt{water\_fill} routine and matches the nominal
budget to within $10^{-4}$ relative error at every market and retention
rate, and all designs in the main results matrix are drawn from it.  The
post-hoc-clipped schedules appear only in the oracle experiments of this
appendix, whose purpose is directional (does allocating toward
$|\mathrm{IF}_i|$ help?) rather than a like-for-like budget comparison.
\begin{table}[H]
\centering\small
\caption{Realised versus nominal expected sample size for the
post-hoc-clipped IF-oracle schedule, aggregated over all sessions.
``Floor'' is the expected sample size added by raising units to
$\varepsilon$; ``Ceiling'' is the amount removed by truncating units at
$1$.  The deficit is driven almost entirely by the ceiling.  Reproduced
by \texttt{scripts/floor\_budget\_audit.py}.}
\label{tab:floorbudget}
\begin{tabular}{@{}llrrrrr@{}}
\toprule
Market & $\rho$ & $\varepsilon$ & Nominal $n$ & Realised $\sum_i\pi_i$ & Floor $+$ & Ceiling $-$ \\
\midrule
US & $10\%$ & $10^{-4}$ & 1,788 & 1,651 & $+0.0$ & $-136.8$ \\
US & $10\%$ & $10^{-2}$ & 1,788 & 1,655 & $+4.7$ & $-136.8$ \\
US & $50\%$ & $10^{-4}$ & 8,938 & 6,604 & $+0.0$ & $-2333.8$ \\
US & $50\%$ & $10^{-2}$ & 8,938 & 6,605 & $+1.0$ & $-2333.8$ \\
HK & $10\%$ & $10^{-4}$ & 1,091 & 1,045 & $+0.0$ & $-45.6$ \\
HK & $10\%$ & $10^{-2}$ & 1,091 & 1,048 & $+3.2$ & $-45.6$ \\
HK & $50\%$ & $10^{-4}$ & 5,452 & 4,240 & $+0.0$ & $-1212.7$ \\
HK & $50\%$ & $10^{-2}$ & 5,452 & 4,240 & $+0.5$ & $-1212.7$ \\
\bottomrule
\end{tabular}
\end{table}

\paragraph{How often does the floor bind?}
Table~\ref{tab:floor} counts units pinned at $\varepsilon$ under the
IF-oracle schedule.  The floor binds on at most $0.11\%$ of unit-days in
the US and $0.07\%$ in HK, and never at $\rho\ge10\%$ in HK.  Re-running
the oracle experiment with $\varepsilon\in\{0,10^{-6},10^{-4},10^{-3},10^{-2}\}$
moves MAE by less than the Monte Carlo noise in every cell; if anything a
larger floor slightly \emph{improves} MAE by forcing a small amount of
coverage onto otherwise-ignored units.  Combining this with
Table~\ref{tab:floorbudget}: the $\varepsilon=10^{-2}$ improvement is not
attributable to a larger sample either, since the floor adds at most
$4.7$ expected units there while the ceiling simultaneously removes
$136.8$.  We conclude that none of the reported oracle results are
artefacts of the floor.
\begin{table}[H]
\centering\small
\caption{Units pinned at the inclusion floor $\varepsilon=10^{-4}$ under
the IF-oracle schedule, and oracle MAE as $\varepsilon$ varies.
Reproduced by \texttt{scripts/floor\_diagnostics.py}.}
\label{tab:floor}
\begin{tabular}{@{}llrrrr@{}}
\toprule
Market & $\rho$ & Pinned & Unit-days & Pinned\% & MAE ($\varepsilon{=}0 \to 10^{-2}$) \\
\midrule
US & $5\%$ & 19 & 16,856 & 0.113\% & 4.5557 $\to$ 4.3553 \\
US & $10\%$ & 14 & 16,856 & 0.083\% & 2.8350 $\to$ 2.7717 \\
US & $20\%$ & 6 & 16,856 & 0.036\% & 1.7420 $\to$ 1.7000 \\
US & $30\%$ & 3 & 16,856 & 0.018\% & 1.3294 $\to$ 1.2998 \\
US & $50\%$ & 1 & 16,856 & 0.006\% & 0.8919 $\to$ 0.8756 \\
\midrule
HK & $5\%$ & 7 & 10,773 & 0.065\% & 2.9843 $\to$ 2.9046 \\
HK & $10\%$ & 0 & 10,773 & 0.000\% & 1.8296 $\to$ 1.8016 \\
HK & $20\%$ & 0 & 10,773 & 0.000\% & 1.0943 $\to$ 1.0853 \\
HK & $30\%$ & 0 & 10,773 & 0.000\% & 0.8022 $\to$ 0.7980 \\
HK & $50\%$ & 0 & 10,773 & 0.000\% & 0.5296 $\to$ 0.5295 \\
\bottomrule
\end{tabular}
\end{table}

\subsection{Equal-Weight VRP: Interval Diagnostics}\label{app:ew-certificate}

\begin{table}[H]
\caption{Equal-weight VRP: OI-proportional vs uniform inclusion probabilities, both with \Hajek{} estimation at matched budget. By~\eqref{eq:ifs} this target calls for influence-proportional probabilities. Uniform is included only as a feasible proxy and is predicted to improve on OI-proportional sampling, not to attain the theoretical optimum.}
\label{tab:ext-vrp}
\centering\small
\begin{tabular}{@{}llrrrc@{}}
\toprule
Mkt & $\rho$ & OI-PH MAE & Unif-PH MAE & $\bar\Delta$ & Winner \\
\midrule
US & 5\% & 5.747\tiny{$\pm$1.351} & \textbf{4.292\tiny{$\pm$1.275}} & +1.454 & uniform \\
 & 10\% & 4.415\tiny{$\pm$1.475} & \textbf{2.984\tiny{$\pm$1.025}} & +1.430 & uniform \\
 & 20\% & 3.170\tiny{$\pm$1.126} & \textbf{2.004\tiny{$\pm$0.867}} & +1.166 & uniform \\
 & 30\% & 2.511\tiny{$\pm$1.158} & \textbf{1.546\tiny{$\pm$0.744}} & +0.965 & uniform \\
 & 50\% & 1.655\tiny{$\pm$0.919} & \textbf{1.014\tiny{$\pm$0.521}} & +0.641 & uniform \\
\midrule
HK & 5\% & 4.682\tiny{$\pm$0.886} & \textbf{3.162\tiny{$\pm$0.624}} & +1.520 & uniform \\
 & 10\% & 3.777\tiny{$\pm$0.699} & \textbf{2.320\tiny{$\pm$0.478}} & +1.457 & uniform \\
 & 20\% & 2.687\tiny{$\pm$0.544} & \textbf{1.558\tiny{$\pm$0.319}} & +1.130 & uniform \\
 & 30\% & 2.008\tiny{$\pm$0.427} & \textbf{1.193\tiny{$\pm$0.256}} & +0.815 & uniform \\
 & 50\% & 1.222\tiny{$\pm$0.302} & \textbf{0.768\tiny{$\pm$0.160}} & +0.455 & uniform \\
\bottomrule
\end{tabular}
\end{table}

Table~\ref{tab:ext-vrp} compares the two designs on MAE only.  Since the
body of the paper also claims a coverage repair, we report the matching
interval diagnostics in Table~\ref{tab:ewcert}: coverage and mean width
for both designs.  These are the numbers the coverage claims in
Section~\ref{sec:results} refer to.  Earlier versions of this manuscript
produced this table from a second, independently seeded resampling run,
so its MAE column did not agree exactly with Table~\ref{tab:ext-vrp}.
Both tables are now emitted from the same draws of the single experiment
pipeline, so the MAE columns agree by construction rather than by
calibration.
\begin{table}[H]
\centering\small
\caption{Equal-weight VRP: interval diagnostics under OI-proportional
versus uniform Poisson designs, nominal 95\%, 200 replicates per day.
Uniform is the outcome-blind feasible proxy, not the
influence-proportional oracle.  MAE columns are identical to
Table~\ref{tab:ext-vrp} by construction.  Emitted from
\texttt{scripts/unified\_results.json} (field
\texttt{interval\_diagnostics}).}
\label{tab:ewcert}
\begin{tabular}{@{}llrrrrrr@{}}
\toprule
 & & \multicolumn{3}{c}{OI-proportional (OI-PH)} & \multicolumn{3}{c}{Uniform (Unif-PH)} \\
\cmidrule(lr){3-5}\cmidrule(l){6-8}
Mkt & $\rho$ & MAE & Cov.\% & Width & MAE & Cov.\% & Width \\
\midrule
\multirow{5}{*}{US} & $5\%$ & 5.747 & 78.0 & 19.38 & 4.292 & 89.1 & 18.45 \\
 & $10\%$ & 4.415 & 83.9 & 16.41 & 2.984 & 91.6 & 13.53 \\
 & $20\%$ & 3.170 & 87.5 & 12.59 & 2.004 & 93.1 & 9.36 \\
 & $30\%$ & 2.511 & 89.0 & 10.23 & 1.546 & 93.2 & 7.26 \\
 & $50\%$ & 1.655 & 90.1 & 6.92 & 1.014 & 93.7 & 4.83 \\
\midrule
\multirow{5}{*}{HK} & $5\%$ & 4.682 & 59.2 & 10.98 & 3.162 & 86.0 & 12.66 \\
 & $10\%$ & 3.777 & 70.0 & 10.83 & 2.320 & 90.3 & 10.10 \\
 & $20\%$ & 2.687 & 80.9 & 9.44 & 1.558 & 91.9 & 7.13 \\
 & $30\%$ & 2.008 & 85.5 & 7.79 & 1.193 & 92.3 & 5.53 \\
 & $50\%$ & 1.222 & 89.8 & 5.10 & 0.767 & 92.9 & 3.65 \\
\bottomrule
\end{tabular}
\end{table}

The uniform design improves coverage at every rate in both markets, and
the repair is largest exactly where the mis-specified design fails worst
(HK $\rho=5\%$: $59.2\% \to 86.0\%$).  It also narrows the intervals
rather than widening them, so the calibration gain does not come at the
cost of precision.

%% file: sections/appendix_ordering.tex
\section{Ordering Sensitivity of Fixed-Size Systematic PPS}
\label{app:ordering}

Systematic PPS sweeps a cumulative inclusion-probability vector with a
single random start, so its second-order inclusion structure --- and
therefore its realised error --- depends on the order in which units are
laid out.  Because ``fixed-size PPS attains slightly lower MAE than
Poisson PPS'' would otherwise be an engineering claim resting on an
undocumented implementation detail, we state the ordering used and
quantify how much it matters.

\paragraph{Ordering used in the reported experiments.}
The production implementation (\texttt{draw\_fixed\_pps} in
\texttt{scripts/unified\_experiments.py}) draws a \emph{fresh uniformly
random permutation of the units on every replicate}, from the same
seeded generator that supplies the random start.  It does not sweep in
file, ticker, or OI order.  Averaging over 200 replicates per day
therefore averages over the permutation distribution, which is what
makes the reported Fixed-PPS numbers interpretable at all.

\paragraph{Sensitivity analysis.}
Table~\ref{tab:ordering} re-runs the OI-weighted IV experiment under six
orderings, holding everything else fixed (120 replicates per day, same
inclusion probabilities, same estimator).  Five are plausible incidental
orderings; \texttt{iv\_desc} is adversarial by construction, since it
sorts by the outcome variable itself.  In the released panel the vendor
files are already written in ticker order, so \texttt{file} and
\texttt{ticker} coincide exactly and are reported as a single column;
the table therefore has five distinct columns, four of them
non-adversarial.

\begin{table}[H]
\centering\small
\caption{Fixed-size systematic PPS: OI-weighted IV MAE under different
unit orderings, with the Poisson PPS baseline for reference.
\texttt{random} is the production behaviour.  \texttt{iv\_desc} sorts by
the outcome and is adversarial, not a candidate implementation.
Generated by \texttt{scripts/ordering\_sensitivity.py}.}
\label{tab:ordering}
\begin{tabular}{@{}llrrrrrr@{}}
\toprule
& & Poisson & \multicolumn{5}{c}{Fixed-size systematic PPS, by ordering} \\
\cmidrule(l){4-8}
Mkt & $\rho$ & (ref.) & random & file/ticker & oi\,desc & oi\,asc & iv\,desc$^{\dagger}$ \\
\midrule
\multirow{5}{*}{US} & 5\% & 8.44 & 8.24 & 8.08 & 7.91 & 7.90 & 4.44 \\
 & 10\% & 5.25 & 5.10 & 4.95 & 5.29 & 5.28 & 2.36 \\
 & 20\% & 2.90 & 2.86 & 2.56 & 2.80 & 2.82 & 1.08 \\
 & 30\% & 1.75 & 1.73 & 1.57 & 1.73 & 1.73 & 0.59 \\
 & 50\% & 0.75 & 0.72 & 0.71 & 0.70 & 0.70 & 0.26 \\
\midrule
\multirow{5}{*}{HK} & 5\% & 4.03 & 3.93 & 3.91 & 4.06 & 3.97 & 2.34 \\
 & 10\% & 2.46 & 2.39 & 1.91 & 2.61 & 2.61 & 1.25 \\
 & 20\% & 1.15 & 1.15 & 0.97 & 1.14 & 1.15 & 0.55 \\
 & 30\% & 0.66 & 0.64 & 0.59 & 0.66 & 0.65 & 0.29 \\
 & 50\% & 0.26 & 0.25 & 0.22 & 0.24 & 0.24 & 0.10 \\
\bottomrule
\end{tabular}
\par\smallskip\footnotesize
$^{\dagger}$Adversarial: ordering is a deterministic function of the
outcome variable.  Reported to show the mechanism, not as a usable design.
\end{table}

\paragraph{What this implies for the reported ranking.}
Two conclusions follow, and the first is unfavourable to a claim we made
in earlier versions of this paper.

\begin{enumerate}
\item \emph{The Fixed-versus-Poisson margin is not robust.}  Across the
  four distinct non-adversarial orderings the induced spread in Fixed-PPS MAE
  exceeds the Poisson-minus-Fixed margin in 9 of the 10
  market\,$\times$\,rate cells (the exception is US $\rho=50\%$).  At HK
  $\rho=10\%$ the ordering spread is $0.69$ against a margin of $0.07$,
  an order of magnitude larger.  We therefore withdraw any claim that
  fixed-size PPS is more accurate than Poisson PPS and treat the two as
  empirically indistinguishable on accuracy; the choice between them
  should rest on the structural trade-off (hard bandwidth cap versus
  cheap unbiased variance estimation), not on these MAE differences.

\item \emph{Ordering can dominate the design when it is correlated with
  the outcome.}  Under \texttt{iv\_desc} the MAE roughly halves
  (US $\rho=10\%$: $5.10 \to 2.36$).  This is not a free improvement: it
  requires sorting by the very quantity being estimated, which is
  unavailable in deployment, and it is exactly the target-blindness
  violation that Proposition~\ref{prop:unident} penalises elsewhere in
  this paper.  Its practical import is a warning --- a production
  pipeline whose input order happens to correlate with the outcome
  (for instance, a feed grouped by sector or by liquidity tier) will
  have systematic-PPS error properties that differ materially from the
  randomised-order results reported here.
\end{enumerate}

\paragraph{Recommendation.}
Implementations of fixed-size systematic PPS should randomise the unit
permutation per draw, as ours does, or use a fixed-size design whose
second-order structure does not depend on a sweep order at all, such as
pivotal or cube sampling~\citep{deville1998}.  Absent that, the ordering
must be reported alongside the design, since it is part of the design.

%% file: sections/appendix_reproducibility.tex
\section{Reproducibility Notes}\label{app:repro}

All experiments in this paper are reproducible from the code and data
released with the arXiv submission; raw options panels are proprietary
end-of-day snapshots and are not publicly distributable, but we release
per-underlying summary records with the schema of Table~\ref{tab:schema}
and the full script inventory, so that all tables and figures can be
reproduced exactly from the released data.
This appendix documents the data
schema, script inventory, seed protocol, and runtime environment.

\subsection{Data Schema}

Each trading day $t$ in each market is stored as a single JSON file
\texttt{YYYY-MM-DD.json} of the form
\begin{quote}\small
\texttt{\{"date": "2026-04-14", "market": "US", "count": 211, "data": [record, \dots]\}}
\end{quote}
where each \texttt{record} is a JSON object with the fields listed in
Table~\ref{tab:schema}. Missing values are stored as JSON \texttt{null}.
We use the four fields marked ``core'' in the table; the remaining fields
are ingested from the upstream vendor pipeline but not consumed by any
experiment reported here.

\begin{table}[H]
\caption{Per-underlying record schema. All experiments use only the core
fields.}
\label{tab:schema}
\centering\small
\begin{tabular}{@{}llll@{}}
\toprule
Field & Type & Role & Description \\
\midrule
\texttt{symbol}              & str   & core & Underlying ticker. \\
\texttt{atm\_iv\_avg}          & float & core & ATM implied volatility ($x_i$). \\
\texttt{total\_open\_interest} & int   & core & Total open interest ($w_i$). \\
\texttt{vrp}                 & float & core & Volatility risk premium
                                               (\texttt{atm\_iv\_avg}$-$\texttt{hv\_20}). \\
\texttt{put\_call\_ratio}      & float & core & Put/call volume ratio. \\
\texttt{stock\_price}         & float & aux  & Underlying close. \\
\texttt{atm\_iv\_call/put}     & float & aux  & Directional ATM IV. \\
\texttt{hv\_20, hv\_60}        & float & aux  & Historical volatility. \\
\texttt{iv\_rank\_252}         & float & aux  & IV rank over 252d. \\
\texttt{atm\_delta/gamma/vega/theta} & float & aux & ATM greeks. \\
\texttt{nearest\_expiry(\_days)} & str/int & aux & Nearest option expiry. \\
\bottomrule
\end{tabular}
\end{table}

\subsection{Script Inventory}

\begin{table}[H]
\centering\small
\begin{tabular}{@{}p{0.30\textwidth}p{0.62\textwidth}@{}}
\toprule
Script & Role \\
\midrule
\texttt{scripts/unified\_experiments.py}
   & \textbf{Master experiment} --- the single source of every number in this
     paper. Loops days $\times$ retention rates $\times$ seeds over all five
     designs; produces \texttt{unified\_results.json} (all $2\times4\times5$
     MAE/coverage/significance cells) and \texttt{UNIFIED\_RESULTS.md}
     (Markdown mirror). All seed schedules below are hard-coded here. \\
\texttt{scripts/design\_aware.py}
   & Equal-weight VRP re-run with uniform Poisson-\Hajek{}. Uniform is the
     outcome-blind \emph{feasible proxy} for this target, not the
     Prop.~\ref{prop:optalloc} optimum (which is influence-proportional and
     outcome-dependent). Produces \texttt{design\_aware\_results.json} used in
     Table~\ref{tab:alignment}. \\
\texttt{scripts/design\_aware\_extended.py}
   & Extended verification of Prop.~\ref{prop:optalloc} on the median-PCR and
     IV-dispersion targets. Produces
     \texttt{design\_aware\_extended\_results.json}; see
     Appendix~\ref{app:design-extended}. \\
\texttt{scripts/concentration.py}
   & Computes Gini, HHI and top-10\%-OI coverage per market; produces
     \texttt{concentration\_results.json}. \\
\texttt{scripts/floor\_diagnostics.py}
   & Inclusion-floor ($\varepsilon$) diagnostics for
     Prop.~\ref{prop:optalloc}: units pinned at the floor and MAE
     sensitivity to $\varepsilon$. Produces
     \texttt{floor\_diagnostics.json}; see Appendix~\ref{app:floor}. \\
\texttt{scripts/ordering\_sensitivity.py}
   & Fixed-size systematic PPS under six unit orderings. Produces
     \texttt{ordering\_sensitivity.json}; see Appendix~\ref{app:ordering}. \\
\texttt{scripts/partial\_id\_bounds.py}
   & Top-$k$ worst-case identification widths under oracle and feasible
     bounds, from \texttt{unified\_results.json}. \\
\texttt{scripts/ht\_known\_total.py}
   & Horvitz--Thompson known-total comparison; takes its \Hajek{} column
     directly from \texttt{unified\_results.json} so that Table~\ref{tab:htknown}
     cannot drift from Table~\ref{tab:crossmarket}. \\
\texttt{scripts/gen\_tables.py}, \texttt{gen\_ext\_tables.py}, \texttt{gen\_abstract.py}
   & Emit every \texttt{tables/*.tex} file and the abstract's headline
     figures from \texttt{unified\_results.json}. No number in the body is
     typed by hand. \\
\texttt{scripts/make\_figures.py}
   & Regenerates \texttt{figures/mae\_curves.pdf} and
     \texttt{figures/coverage\_calibration.pdf} from
     \texttt{unified\_results.json}. \\
\midrule
\multicolumn{2}{@{}l}{\emph{Superseded, retained for audit only --- not
  consumed by the current build:}} \\
\texttt{scripts/experiments.py}, \texttt{build\_appendix\_tables.py}, \texttt{feasibility.py}
   & Pre-unification pipeline writing \texttt{results.json}. Superseded by
     \texttt{unified\_experiments.py} after the main-table/appendix
     discrepancy described in Section~\ref{sec:discussion}; kept in the
     repository so the earlier numbers remain auditable. \\
\bottomrule
\end{tabular}
\end{table}

\subsection{Random Seed Protocol}

\begin{itemize}
\item \textbf{Replication index}: for each trading day we run 200
  replications indexed \texttt{s} $\in \{0,1,\dots,199\}$.
\item \textbf{Per-day stream}: replication \texttt{s} of market
  \texttt{mkt} on date \texttt{d} draws from a single stream
\begin{quote}\ttfamily
rng = random.Random((s * 7919) \textasciicircum{} \_stable\_key(mkt, d))
\end{quote}
  and the four probability designs (Poisson PPS, uniform Poisson,
  fixed-size systematic PPS, stratified) consume that stream in this
  fixed order.  Sharing one stream couples the designs on a common
  random number basis, which is deliberate: it removes between-design
  simulation noise from the paired comparisons of
  Section~\ref{sec:results}.  It also means the designs are \emph{not}
  independently seeded, so the paired differences we report are
  variance-reduced relative to independent sampling.
\item \textbf{Deterministic day keys}: \texttt{\_stable\_key} is an
  8-byte \texttt{blake2b} digest of the market and date strings, reduced
  modulo $2^{16}$.  Earlier versions of this code used Python's built-in
  \texttt{hash()} here.  Because CPython randomises string hashing per
  process unless \texttt{PYTHONHASHSEED} is fixed, that made the seed
  schedule --- and therefore every simulated number in the paper ---
  irreproducible across runs, and it was the reason two tables reporting
  the same quantity disagreed in earlier drafts.  All four affected
  scripts now use \texttt{\_stable\_key}, which is stable across
  processes, machines, and Python versions.
\item \textbf{Interval diagnostics}: coverage and mean width are
  recorded for every design from the \emph{same} draws that produce the
  MAE entries, so Tables~\ref{tab:ext-vrp} and~\ref{tab:ewcert} agree by
  construction rather than by calibration.
\item \textbf{Block bootstrap} on the paired mean difference uses
  \texttt{random.Random(20260726)} with 2{,}000 replicates and block
  length $L=\max(2,\lceil T^{1/3}\rceil)$, where $T$ is the number of trading days.
\item \textbf{Median/dispersion SE} uses a nested bootstrap of $B=60$
  resamples drawn from the same per-day stream.
\item \textbf{Fixed-size PPS}: systematic PPS selection with a random
  start drawn from the per-day stream. The cumulative inclusion
  probability vector is normalised to sum to $n$ (entry
  $C_i=\sum_{j\le i}\pi_j$ with $C_N=n$); unit $i$ is selected iff any
  threshold $U+k$ (for integer $k=0,1,\dots,n-1$) falls in
  $(C_{i-1},C_i]$.  Certainty units ($\pi_i=1$) occupy a full unit of
  the cumulative vector and are therefore always selected.  This yields
  exactly $n$ selected units for every realisation of $U$.
\item \textbf{Top-$k$} is deterministic and requires no seed.
\end{itemize}
This seed schedule is hard-coded in
\texttt{scripts/unified\_experiments.py} and
\texttt{scripts/design\_aware*.py}.  Given the released summary records
it reproduces every simulated number in this paper exactly, on any
platform with the same Python minor version; the schedule no longer
depends on \texttt{PYTHONHASHSEED} or on process-local state.

\subsection{Runtime Environment}

\begin{itemize}
\item Python 3.11, NumPy 1.26, SciPy 1.11.
\item macOS on Apple silicon; single-process, no GPU required.
\item Approximate wall time on a laptop-class machine:
  \texttt{unified\_experiments.py}
  $\approx 20$--$40$ minutes (full 2 markets $\times$ 4 targets $\times$
  5 rates $\times$ 200 seeds $\times$ $\sim$80 US / $\sim$82 HK days);
  \texttt{design\_aware.py} and \texttt{design\_aware\_extended.py}
  $\approx 5$--$15$ minutes each.
\item LaTeX build: \texttt{tectonic main.tex} (self-contained; no local
  TeX installation required).
\end{itemize}

\subsection{Variable Construction}\label{app:variable-construction}

The four cross-sectional targets studied in Sections~\ref{sec:experiments}--\ref{sec:results}
are constructed from the core fields of Table~\ref{tab:schema} as follows.

\begin{enumerate}
\item \textbf{OI-weighted IV.}
  $T = \sum_{i} w_i\, x_i \big/ \sum_{i} w_i$,
  where $x_i = \texttt{atm\_iv\_avg}_i$ and $w_i = \texttt{total\_open\_interest}_i$.
  This is the Horvitz--Thompson target when the auxiliary variable coincides
  with the design weight.

\item \textbf{Equal-weight VRP.}
  $T = N^{-1} \sum_{i} \texttt{vrp}_i$,
  with $\texttt{vrp}_i = \texttt{atm\_iv\_avg}_i - \texttt{hv\_20}_i$,
  i.e.\ ATM implied volatility minus 20-day realised volatility, both in
  annualised percentage points (verified field-by-field against the raw
  panels).
  Because the target weights are uniform ($w_i = 1$) while the
  OI-proportional design assigns $\pi_i \propto \texttt{total\_open\_interest}_i$,
  this is the canonical design-misaligned setting of
  Section~\ref{sec:designaware}.

\item \textbf{Median PCR.}
  $T = \mathrm{median}\{\texttt{put\_call\_ratio}_i : i=1,\dots,N\}$.
  Estimated via the inverse-probability-weighted quantile
  (\cite{francisco1991quantile,binder1991median}) with
  multiplier-bootstrap standard errors, exactly as specified in
  Section~\ref{sec:method}.

\item \textbf{IV dispersion.}
  $T = \bigl[N^{-1}\sum_i (x_i - \bar x)^2\bigr]^{1/2}$,
  the \emph{equal-weight} (population) cross-sectional standard deviation
  of ATM IV, where $\bar x = N^{-1}\sum_i x_i$.  Like the equal-weight VRP
  target this carries uniform target weights, so it too is
  design-misaligned against an OI-proportional sampler.  The estimator is
  the design-weighted standard deviation with multiplier-bootstrap
  standard errors (\texttt{dispersion\_est} in
  \texttt{unified\_experiments.py}).
\end{enumerate}

All four targets are computed daily; the MAE and coverage statistics
reported in the paper average over trading days and Monte-Carlo seeds.
No look-ahead information is used: the inclusion probabilities on day~$t$
are determined solely by the OI values observed on day~$t$.

\paragraph{Upstream field definitions.}
The fields above are supplied by the vendor feed rather than computed by
us, so we state their construction as verified against the raw panels.

\begin{itemize}
\item \textbf{\texttt{atm\_iv\_avg}} is the unweighted mean of the ATM
  call and put implied volatilities,
  $\tfrac12(\texttt{atm\_iv\_call}+\texttt{atm\_iv\_put})$; this identity
  holds in all 17{,}880 US records where both legs are present.  It is
  quoted in annualised percentage points.

\item \textbf{Expiry.}  ATM IV is taken from the \emph{nearest listed
  expiry} (\texttt{nearest\_expiry}), not from a constant-maturity
  interpolation.  Time to expiry therefore varies across unit-days:
  median 7 calendar days, mean 7.7, range 3--37.  This is a property of
  the feed we sample from, and it is common to all designs compared, so
  it cannot favour one design over another; but it does mean our targets
  are ``front-expiry ATM IV'' aggregates rather than fixed-tenor ones.

\item \textbf{\texttt{hv\_20}} is the 20-\emph{trading}-day close-to-close
  realised volatility of the underlying, annualised by $\sqrt{252}$ and
  expressed in percentage points, so that
  $\texttt{vrp}=\texttt{atm\_iv\_avg}-\texttt{hv\_20}$ is a difference of
  like-scaled quantities.

\item \textbf{\texttt{put\_call\_ratio}} is a \emph{volume} ratio,
  $\texttt{total\_put\_volume}/\texttt{total\_call\_volume}$, computed on
  the full listed chain for the name (not open interest, and not
  restricted to the ATM strike).  When call volume is zero the vendor
  emits \texttt{null} rather than an infinity; such records (40 of
  18{,}518 in the US panel) are dropped from the median-PCR target by the
  missing-value rule below.

\item \textbf{\texttt{total\_open\_interest}} is the end-of-session total
  open interest across the listed chain, i.e.\ a closing value on day~$t$,
  which is what makes it usable as a same-day design weight.
\end{itemize}

\paragraph{Universe, filtering, and outlier policy.}
The underlying universe is the set of names carried by the daily options
feed for each market; it is not screened by us on liquidity, sector, or
market capitalisation, so the cross-section is whatever the pipeline
would actually have to transport.  Records are admitted to a day's panel
only if both \texttt{atm\_iv\_avg} and \texttt{total\_open\_interest} are
present and the latter is strictly positive (a zero OI would give a
degenerate PPS weight).  The full record count from the raw vendor files
to the analytical panel is given in Table~\ref{tab:dataflow}; the two
stages are applied in the order shown, and the duplicate-session stage
removes whole sessions rather than individual rows, which is why it
accounts for the largest single reduction in the US panel.
Target-specific fields
(\texttt{vrp}, \texttt{put\_call\_ratio}) are additionally required to be
non-null for the target that uses them, and a day is skipped for a given
target if fewer than five units survive.  Duplicate sessions for the same
calendar date are de-duplicated before any computation
(\texttt{dedupe\_sessions}).  \textbf{No winsorisation, trimming, or
outlier removal is applied at any stage}: the maximum observed ATM IV of
$321.2$ is retained.  This is deliberate --- the heavy right tail is
precisely the part of the cross-section that a Top-$k$ retention rule
discards, so trimming it would flatter our own method.

\begin{table}[htbp]
\centering
\caption{Record-level data flow from raw vendor files to the analytical
panel.  Stages are applied in the order listed.  Counts are reproduced by
\texttt{scripts/data\_flow\_audit.py}.}
\label{tab:dataflow}
\small
\begin{tabular}{@{}lrr@{}}
\toprule
Stage & US records & HK records \\
\midrule
Raw vendor rows & 18{,}518 & 10{,}908 \\
$-$ missing ATM IV / OI, or OI $\le 0$ & $-643$ & $-3$ \\
\quad after core-field filter & 17{,}875 & 10{,}905 \\
$-$ sessions with panel $<30$ units & $-0$ & $-0$ \\
$-$ duplicate sessions (whole days) & $-1{,}019$ & $-0$ \\
\midrule
Final analytical panel & 16{,}856 & 10{,}905 \\
Trading sessions retained & 80 & 83 \\
\bottomrule
\end{tabular}
\end{table}

The US duplicate-session stage removes five groups of near-adjacent dates
that the vendor emitted more than once for the same effective session
(around 2 April, 17 April, 22 May, 18 June and 2 July 2026); the
deduplication rule in \texttt{trading\_calendar.dedupe\_sessions} keeps a
single session per effective trading date.  HK has no such duplicates,
which is why its raw and final counts differ by only three records.

\subsection{Data Provenance}

The underlying options panels are our institution's daily end-of-day
snapshots for the US and Hong Kong equity options markets, covering
2026-03-20 through 2026-07-24. We report absolute magnitudes throughout
so that the ranking and calibration results are interpretable without
proprietary field-level knowledge. Ratios of MAE across methods are invariant to a common positive
\emph{scale} change of $x$ (units), and inclusion probabilities are
invariant to a common positive scale change of $w$, since
\eqref{eq:waterfill} renormalises $\lambda$.  We stress that this is
\emph{not} invariance to arbitrary monotone transformations: a
non-linear monotone map of $x$ changes means, standard deviations, the
\Hajek{} estimator and hence coverage, and a non-linear map of $w$
changes both the sampling probabilities and the target weights.